\documentclass[12pt,onecolumn,twoside]{IEEEtran}

\usepackage{calc}

\usepackage{multirow}
\usepackage{cite}
\usepackage{graphicx,subfigure}
\usepackage{psfrag}
\usepackage{amsmath,amssymb,amsthm}
\usepackage{color}
\usepackage{tikz} 
\usepackage{bm}
\usepackage{bbm}
\usepackage{verbatim}
\usepackage[T1]{fontenc}

\usepackage{cases}
\usepackage[noend]{algpseudocode}

\usepackage{xcolor}
\usepackage[most]{tcolorbox}
\usepackage{bm}

\usepackage{tabu}
\usepackage{mathtools}
\usepackage{xifthen}
\usepackage{booktabs}
\usepackage{algorithm}
\usepackage{algpseudocode}
\usepackage{arydshln}

\makeatletter
\def\BState{\State\hskip-\ALG@thistlm}
\makeatother

\newtheorem{theorem}{Theorem}

\newtheorem{lemma}{Lemma}
\newtheorem{definition}{Definition}

\newcommand{\bit}{\begin{itemize}}
\newcommand{\eit}{\end{itemize}}

\newcommand{\bc}{\begin{center}}
\newcommand{\ec}{\end{center}}

\newcommand{\ba}{\begin{array}}
\newcommand{\ea}{\end{array}}

\newcommand{\beq}{\begin{equation}}
\newcommand{\eeq}{\end{equation}}

\newcommand{\beqn}{\begin{equation*}}
\newcommand{\eeqn}{\end{equation*}}

\newcommand{\bean}{\begin{eqnarray*}}
\newcommand{\eean}{\end{eqnarray*}}
\newcommand{\bea}{\begin{eqnarray}}
\newcommand{\eea}{\end{eqnarray}}

\def\hv{\boldsymbol{h}}

\def\wv{\boldsymbol{w}}
\def\xv{\boldsymbol{x}}

\newcommand{\Ac}{{\mathcal A}}
\newcommand{\Bc}{{\mathcal B}}

\newcommand{\Fc}{{\mathcal F}}

\newcommand{\Sc}{{\mathcal S}}

\newcommand{\T}{{\scriptscriptstyle\mathsf{T}}}

\algnewcommand{\IfThenElse}[3]{  \State \algorithmicif\ #1\ \algorithmicthen\ #2\ \algorithmicelse\ #3}

\newcommand{\non}{\nonumber}

\newcommand{\Lh}{\ell}
\newcommand{\Me}{\wv}
\newcommand{\Ry}{\mathrm{s}}

\newcommand{\Vr}{\mathrm{v}} 
\newcommand{\gn}{\eta}

\newcommand{\cb}{c} 
\newcommand{\Ss}{\mathbb{S}}

\newcommand{\Lin}{l}
\newcommand{\Lk}{\mathrm{u}}

\newcommand{\rtuple}{)} 
\newcommand{\ltuple}{(}

\newcommand{\AND}{\land}    
\newcommand{\OR}{\lor}

\newcommand{\MVBAInputMsg}{\wv}

\newcommand{\EncodedSymbol}{y}

\newcommand{\Field}{\mathbb{F}}

\newcommand{\ABBA}{\mathsf{ABBA}}

\newcommand{\thisnodeindex}{i}                       
\newcommand{\send}{\textbf{send}}

\newcommand{\ABAoutput}{b}    
                    
\newcommand{\Output}{\textbf{output}}                      
\newcommand{\terminate}{\textbf{terminate}}

\newcommand{\wait}{\textbf{wait}}                             
                        
\newcommand{\ECC}{\mathrm{ECC}}

\newcommand{\ECCEnc}{\mathrm{ECCEnc}}   
\newcommand{\ECCDec}{\mathrm{ECCDec}} 
\newcommand{\Alphabet}{\Field_{\alphabetsize}}      
        
\newcommand{\MVBAOutputMsg}{\hat{\wv}}

\newcommand{\OECsymbolset}{\mathbb{Z}_{\mathrm{oec}}}   
\newcommand{\Lkset}{\mathbb{U}}   
\newcommand{\SIone}{\mathsf{SI1}}                         
\newcommand{\SItwo}{\mathsf{SI2}}

\newcommand{\defaultvalue}{\bot}      
\newcommand{\ECCEncindicator}{I_\mathrm{ecc}}

\newcommand{\Phthreeindicator}{I_\mathrm{3}}

\newcommand{\OECCorrectSymbolSet}{\mathbb{Y}_{\mathrm{oec}}}

\newcommand{\OEC}{\mathrm{OEC}}

\newcommand{\Pass}{\textbf{pass}}

\newcommand{\DM}{\mathrm{DM}}

\newcommand{\networksizen}{n}                                           
\newcommand{\kencode}{k}                                            
\newcommand{\networkfaultsizet}{t}

\newcommand{\successindicator}{\mathrm{s}}

\newcommand{\IDCOOL}{\mathrm{ID}}

\newcommand{\BA}{\mathsf{BA}}

\newcommand{\BBA}{\mathsf{BBA}}

\newcommand{\OciorRBC}{\mathsf{OciorRBC}}

\newcommand{\ABA}{\mathrm{ABA}}

\newcommand{\BUA}{\mathrm{UA}}

\newcommand{\RBA}{\mathsf{RBA}}      
\newcommand{\BRBA}{\mathsf{BRBA}}       
        
\newcommand{\HMDM}{\mathrm{HMDM}}

\newcommand{\RBC}{\mathrm{RBC}}

\newcommand{\OECsymbol}{z}                                                                            
\newcommand{\OECsymbolsetInitial}{\mathbb{Z}_{\mathrm{oec}}}     
\newcommand{\OECSI}{I_\mathrm{oec}}                                                                        
     
\newcommand{\OECSIFinal}{I_\mathrm{oecfinal}}                                                                              
         
\newcommand{\VrOutput}{\mathrm{v}^\diamond} 
\newcommand{\OciorRBA}{\mathsf{OciorRBA}}            
              
\newcommand{\deliver}{\textbf{deliver}}                     
\newcommand{\SHMDM}{\mathsf{SHMDM}}

\newcommand{\alphabetsize}{q}

\newcommand{\OciorACOOL}{\mathsf{OciorACOOL}}           
\newcommand{\COOLBUA}{\mathsf{COOL}\text{-}\mathsf{UA}}    
\newcommand{\COOLHMDM}{\mathsf{COOL}\text{-}\mathsf{HMDM}}                                                                      
\newcommand{\SYMBOL}{\mathsf{SYMBOL}}                         
\newcommand{\defeqnew}{:=}                         
\newcommand{\READY}{\mathsf{READY}}                            
\newcommand{\CORRECTSYMBOL}{\mathsf{CORRECT}}                         
\newcommand{\LEADER}{\mathsf{LEADER}}                            
\newcommand{\INITIAL}{\mathsf{INITIAL}}

\newcommand{\MeNew}{\bar{\Me}}
\newcommand{\RyNew}{\bar{\Ry}}
\newcommand{\VrNew}{\bar{\Vr}}  
\newcommand{\ysymbolNew}{\bar{y}}  
\newcommand{\OECCorrectSymbolSetNew}{\bar{\mathbb{Y}}_{\mathrm{oec}}}    
\newcommand{\ystrongmajority}{\check{y}}  
\newcommand{\SsNew}{\bar{\Ss}}
\newcommand{\MeSecond}{\tilde{\Me}}
\newcommand{\RySecond}{\tilde{\Ry}}
\newcommand{\VrSecond}{\tilde{\Vr}}  
\newcommand{\ysymbolSecond}{\tilde{y}}  
\newcommand{\SsSecond}{\tilde{\Ss}}
\newcommand{\ABBAOutput}{\Vr^{\star}}   
\newcommand{\NEWSYMBOL}{\mathsf{NewSYMBOL}}    
   
\newcommand{\ySymbolSet}{\mathbb{M}}   
\newcommand{\OciorCOOL}{\mathsf{OciorCOOL}}             
\newcommand{\COOL}{\mathsf{COOL}}                   
                   
\newcommand{\ECCEncodingFunction}{f}                   
\newcommand{\Meg}{\bar{\wv}}   
\newcommand{\AsetNOzero}{\ddot{\Ac}}   
\newcommand{\BsetNOone}{\ddot{\Bc}}    
\newcommand{\gnNozero}{\ddot{\eta}}  
\newcommand{\iprime}{i_1}  
\newcommand{\iprimeprime}{i_2}  
\newcommand{\ithree}{i_3}  
\newcommand{\LksetHonest}{\mathbb{V}}   
   
\newcommand{\nbar}{\bar{n}}   
             
\newcommand{\OciorACOOLsmallt}{\mathsf{OciorACOOL}^\star}            
\newcommand{\LEADERMESSAGE}{\mathsf{MESSAGE}}

\algdef{SE}[SUBALG]{Indent}{EndIndent}{}{\algorithmicend\ }%
\algtext*{Indent}
\algtext*{EndIndent}

\begin{document}
\sloppy
\title{COOL Is Optimal in Error-Free Asynchronous Byzantine Agreement}

\author{Jinyuan Chen 
}

\maketitle
\pagestyle{headings}

\begin{abstract}
$\COOL$ (Chen'21) is an error-free, information-theoretically secure Byzantine agreement ($\BA$) protocol proven to achieve $\BA$ consensus in the synchronous setting for an $\Lh$-bit message, with a total communication complexity of $O(\max\{n\Lh, nt \log \alphabetsize\})$ bits, four communication rounds in the worst case, and a single invocation of a binary $\BA$, under the optimal resilience assumption $n \geq 3t + 1$ in a network of $n$ nodes, where up to $t$ nodes may behave dishonestly. Here, $\alphabetsize$ denotes the alphabet size of the error correction code used in the protocol.

In this work, we present an adaptive variant of $\COOL$, called $\OciorACOOL$, which achieves error-free,  information-theoretically secure $\BA$ consensus in the asynchronous setting with total $O(\max\{n\Lh, n t \log \alphabetsize \})$  communication bits, $O(1)$  rounds, and a single invocation of an asynchronous binary $\BA$ protocol, still under the optimal resilience assumption $n \geq 3t + 1$. Moreover, $\OciorACOOL$ retains the same low-complexity, traditional $(n, k)$ error-correction encoding and decoding as $\COOL$, with $k=t/3$.  
\end{abstract}


 \section{Introduction}

Byzantine agreement ($\BA$) has been extensively studied for over forty years~\cite{PSL:80}. 
In a $\BA$ problem, $n$ consensus nodes aim to agree on a common $\Lh$-bit message, where up to $t$ of the nodes may be dishonest. 
$\BA$ is widely considered as a fundamental building block of Byzantine fault-tolerant distributed systems and cryptographic protocols~\cite{PSL:80, LSP:82, ChenDISC:21, Chen:2020arxiv, ChenOciorCOOL:24, ChenOciorMVBA:24, ChenOciorABA:25, ChenOcior:25, LCabaISIT:21, ZLC:23, FH:06, LV:11, GP:20, LDK:20, NRSVX:20, Patra:11, CT:05,CDG+:24}.    

In the study of multi-valued, error-free $\BA$, significant efforts have been made to improve communication complexity, round complexity, and resilience~\cite{LV:11, GP:20, LDK:20, NRSVX:20, ChenDISC:21, Chen:2020arxiv, ChenOciorCOOL:24}. 
In this direction, Chen proposed the $\COOL$ (or $\OciorCOOL$) protocol, which achieves multi-valued, error-free, information-theoretically secure $\BA$ consensus in the synchronous setting, with a communication complexity of $O(\max\{n\Lh, nt \log \alphabetsize\})$ bits and with four communication rounds in the worst case and a single invocation of a binary $\BA$, under the optimal resilience condition $n \geq 3t + 1$~\cite{ChenDISC:21, Chen:2020arxiv, ChenOciorCOOL:24}.  
Here, $\alphabetsize$ denotes the alphabet size of the error correction code used in the protocol.

The $\COOL$ protocol introduced two key primitives: \emph{Unique Agreement} ($\BUA$) and \emph{Honest-Majority Distributed Multicast} ($\HMDM$).

\emph{Unique Agreement:}   $\BUA$ is a variant of $\BA$. In $\BUA$, each node~$i$ inputs an initial value $\MVBAInputMsg_i$ and seeks to produce an output of the form $(\MVBAInputMsg, \successindicator, \Vr)$, where $\successindicator \in \{0,1\}$ is a success indicator and $\Vr \in \{0,1\}$ is a vote. A $\BUA$ protocol guarantees three properties. One property of $\BUA$ is that if two honest nodes output $(\wv', 1, *)$ and $(\wv'', 1, *)$, respectively, then $\wv' = \wv''$ (\emph{Unique Agreement}). Additionally, if an honest node outputs $(*, *, 1)$, then at least $t+1$ honest nodes eventually output $(\wv, 1, *)$ for the same $\wv$ (\emph{Majority Unique Agreement}). Furthermore, if all honest nodes start with the same input value $\wv$, then all honest nodes eventually output $(\wv, 1, 1)$ (\emph{Validity}).

\emph{Honest-Majority Distributed Multicast:} 
In the $\HMDM$ problem, at least $t+1$ honest nodes act as senders, each multicasting a message to all $n$ nodes. 
The $\HMDM$ property ensures that if all honest senders input the same message $\wv$, then every honest node eventually outputs $\wv$.

In this work, we present an adaptive variant of $\COOL$, called $\OciorACOOL$, which achieves $\BA$ consensus in the asynchronous setting with total communication 
$O(\max\{ n\Lh, n t \log \alphabetsize \})$ bits, $O(1)$ rounds, and a single invocation of an asynchronous  binary $\BA$ protocol,   under the optimal resilience assumption $n \geq 3t + 1$. 
Moreover, $\OciorACOOL$ retains the same low-complexity, traditional $(n, k)$ error-correction encoding and decoding as $\COOL$, with $k = t/3$.

 As shown in Fig.~\ref{fig:OciorCOOLProtocol}, the $\COOL$ protocol is composed of three main components: $\COOLBUA$, a binary $\BA$ ($\BBA$), and $\COOLHMDM$. 
$\OciorACOOL$ extends $\COOL$ to the asynchronous setting. 
As shown in Fig.~\ref{fig:OciorACOOLProtocol}, the $\OciorACOOL$ protocol consists of $\COOLBUA[1]$, $\COOLBUA[2]$, an asynchronous binary $\BA$ ($\ABBA$), a binary reliable Byzantine agreement ($\BRBA$), and $\COOLHMDM[2]$. 
Some existing $\ABBA$ protocols require each honest node to have an input value before producing an output. 
In such cases, $\BRBA$ is employed to ensure the \emph{Totality} property: if one honest node outputs a value, then every honest node eventually outputs a value. 
If the invoked $\ABBA$ protocol already guarantees the \emph{Totality} property, then $\BRBA$ is not required in $\OciorACOOL$.

 \begin{figure}
\centering
\includegraphics[width=9cm]{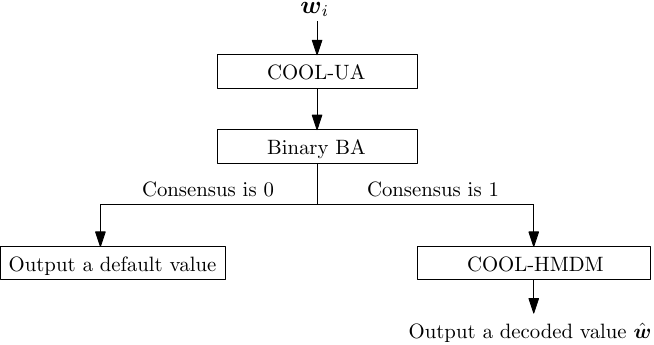}
\caption{A block diagram of the   $\COOL$ protocol, which consists of the $\COOLBUA$, binary $\BA$,  and $\COOLHMDM$ algorithms. 
}
\label{fig:OciorCOOLProtocol}
\end{figure}

 \begin{figure} 
\centering
\includegraphics[width=9cm]{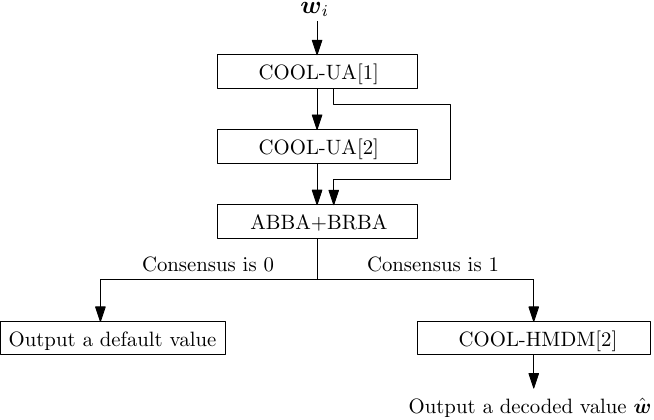}
\caption{A block diagram of the   $\OciorACOOL$ protocol, which consists of the  $\COOLBUA[1]$, $\COOLBUA[2]$, asynchronous  binary $\BA$  ($\ABBA$),   binary reliable Byzantine agreement  ($\BRBA$),  and $\COOLHMDM[2]$ algorithms.  If an invoked $\ABBA$ algorithm  already guarantees  the \emph{Totality} property, then $\BRBA$  is not required  in $\OciorACOOL$. 
}
\label{fig:OciorACOOLProtocol}
\end{figure}

In the direction of error-free asynchronous $\BA$ ($\ABA$), Li and Chen proposed an $\ABA$ protocol built on $\COOL$ that achieves communication complexity $O(\max\{ n\Lh, n t \log \alphabetsize \})$ and expected $O(1)$ rounds, but with a weaker resilience requirement $n \geq 5t + 1$ \cite{LCabaISIT:21}. 
Chen later proposed a multi-valued validated Byzantine agreement protocol with expected communication complexity $O(n\Lh \log n + n^2 \log \alphabetsize)$ bits,   and expected $O(\log n)$ rounds under the optimal resilience   $n \geq 3t + 1$~\cite{ChenOciorMVBA:24}. 
Another $\ABA$ protocol by Chen  achieves expected communication complexity $O(n\Lh + n^3 \log n)$ bits and expected round complexity $O(1)$,   under optimal resilience $n \geq 3t + 1$ \cite{ChenOciorABA:25}. 
Erbes and Wattenhofer proposed an $\ABA$ protocol with communication complexity 
$O\!\left(\frac{n\Lh}{\min\{1, \epsilon^2\}} + n^2 \max\{1, \log \frac{1}{\epsilon}\}\right)$ bits
and near-optimal resilience $t < \frac{n}{3 + \epsilon}$ for some $\epsilon > 0$ \cite{EW:25}.

Recently, Abraham and Asharov~\cite{AA:25} built upon $\COOL$ by incorporating list decoding \cite{Sudan:97,GW:13} for an $(n, k')$ error correction code with $k' = t/7$, achieving asynchronous $\BA$ consensus with a communication complexity of $O(\max\{ n\Lh, n^2 \log n \})$ bits and a single invocation of a binary $\BA$ protocol. 
However, compared to traditional unique decoding of error correction codes, list decoding is less mature in terms of practical implementations. Therefore, off-the-shelf error correction code decoders cannot be directly used for the protocol in \cite{AA:25}.

In contrast, $\OciorACOOL$ achieves error-free $\ABA$ without relying on list decoding. 
It retains the same low-complexity, traditional $(n, k)$ error-correction encoding and decoding as $\COOL$, with $k = t/3$.  
Compared to the protocol in~\cite{AA:25}, $\OciorACOOL$ achieves a smaller communication cost, providing at least a $57\%$ reduction in total communication.

\subsection{Primitives}

 \vspace{.1 in}

  \noindent {\bf Information-Theoretic (IT) Protocol.} 
A protocol that guarantees all required properties without relying on any cryptographic assumptions--such as digital signatures or hash functions--except for the common coin or binary $\BA$ assumptions, is said to be \emph{information-theoretically secure}.  

\medskip
\noindent {\bf Error-Free Protocol.} 
Under the common coin or binary $\BA$ assumptions, a protocol that guarantees all required properties in \emph{every} execution is said to be \emph{error-free}.

\vspace{.1 in} 

\noindent {\bf Error Correction Code ($\ECC$).} 
An $(n, k)$ error correction code consists of an encoding function 
\(\ECCEnc: \Alphabet^{k} \to \Alphabet^{n}\) and a decoding function 
\(\ECCDec: \Alphabet^{n^\diamond} \to \Alphabet^{k}\), where \(\Alphabet\) denotes the alphabet of each symbol, and \(n^\diamond \le n\) represents the number of symbols available for decoding.  
Specifically, the encoding  $[\EncodedSymbol_1, \EncodedSymbol_2, \ldots, \EncodedSymbol_n] \gets \ECCEnc(n, k, \MVBAInputMsg)$
outputs $n$ encoded symbols, where $\EncodedSymbol_j(\MVBAInputMsg)$ denoted  the $j$-th symbol encoded from the message $\MVBAInputMsg$.   
An $(n, k)$ Reed-Solomon (RS) code can correct up to $t$ Byzantine (arbitrary) errors in $n^\diamond$ observed symbols, provided that $2t+  k \leq   n^{\diamond}$  and  $n^\diamond \leq  n$. 
The RS code operates over a  finite field $\Alphabet$, subject to the constraint $n \leq \alphabetsize  -1$ (cf.~\cite{RS:60}), where  $\alphabetsize$ denotes the alphabet size.  
 Since RS codes are limited by the constraint $n \leq \alphabetsize - 1$, alternative error correction codes with constant alphabet size, such as expander codes~\cite{SS:96}, can be used instead.

\vspace{.1 in} 

 \noindent  {\bf Online Error  Correction  ($\OEC$).}     
 Online error correction is particularly useful for decoding messages in asynchronous settings~\cite{BCG:93}. 
A node may not be able to decode  the message from $n^\diamond$ symbol observations,   when the actual number of Byzantine errors among these observations, denoted by $t^\diamond$, satisfies 
$2t^\diamond + k > n^\diamond$.  
In such cases, the node waits for an additional symbol before attempting to decode again. This procedure repeats until the message is successfully reconstructed. 
In the worst case, $\OEC$ may perform up to $t$ such trials before decoding the message.

\begin{definition}[{\bf Byzantine Agreement ($\BA$)}]
In a $\BA$ protocol, consensus nodes aim to reach an agreement on a common value. 
It guarantees three properties:
\begin{itemize}
\item {\bf Termination:} If all honest nodes have received their inputs, then every honest node eventually produces an output and terminates.
\item {\bf Consistency:} If any honest node outputs a value $\wv$, then all honest nodes eventually output the same value $\wv$.
\item {\bf Validity:} If all honest nodes start with the same input value $\wv$, then every honest node eventually outputs $\wv$.
\end{itemize}
\end{definition}

 \begin{definition} [{\bf Reliable Byzantine Agreement ($\RBA$)}]
$\RBA$  is a  variant    of both the Byzantine agreement  and reliable broadcast problems. 
An $\RBA$ protocol guarantees the following  three properties:
\begin{itemize}
\item  {\bf Consistency:} If any two honest nodes output $\wv'$ and $\wv''$, respectively, then  $\wv'=\wv''$.
\item   {\bf Validity:} Same as in $\BA$.   
\item  {\bf Totality:}  If one  honest node outputs a value, then every honest node  eventually outputs a value.  
\end{itemize} 
\end{definition}  
The proposed $\OciorACOOL$ uses a binary $\RBA$ ($\BRBA$), which is described in  Lines~\ref{line:OciorACOOLRBAbegin}-\ref{line:OciorACOOLph3} of Algorithm~\ref{algm:OciorACOOL}. 
In Algorithm~\ref{algm:OciorRBA}, we provide an optimal, error-free, multi-valued $\RBA$ protocol, denoted as $\OciorRBA$, which is an updated version of Algorithm~4 in~\cite{ChenOciorMVBA:24}.
The $\OciorRBA$ protocol achieves error-free multi-valued $\RBA$ consensus with a total communication complexity of $O(\max\{n\Lh, n^2 \log \alphabetsize \})$ bits, requiring at most \emph{five} asynchronous communication rounds in the worst case (\emph{four} rounds in the good case, when all honest nodes have the  same input message), under the optimal resilience condition $n \geq 3t + 1$. 
The proof of $\OciorRBA$ is similar to that of the $\OciorRBC$ protocol (see the proof of $\OciorRBC$ in~\cite{ChenOciorCOOL:24}).

 \begin{definition}[{\bf Reliable Broadcast ($\RBC$)}]
In the reliable broadcast problem, a designated leader disseminates an input message to $n$ nodes. 
A protocol implementing $\RBC$ must satisfy the following properties:
\begin{itemize}
\item {\bf Consistency:} Same as in $\RBA$.     
\item {\bf Validity:} If the leader is honest and broadcasts a value $\wv$,   every honest node eventually outputs $\wv$.
\item {\bf Totality:} Same as in $\RBA$.   
\end{itemize}
\end{definition}
In Algorithm~\ref{algm:OciorRBC}, we present an optimal, error-free, multi-valued, balanced $\RBC$ protocol, denoted as $\OciorRBC$, which is an updated version of Algorithm~3 in~\cite{ChenOciorCOOL:24}.
The $\OciorRBC$ protocol achieves error-free multi-valued $\RBC$ consensus with a total communication complexity of $O(\max\{n\Lh, n^2 \log \alphabetsize\})$ bits (a per-node communication complexity of $O(\max\{\Lh, n \log \alphabetsize\})$ bits), requiring at most \emph{six} asynchronous communication rounds in the worst case (or \emph{five} rounds in the good case, without balanced communication), under the optimal resilience condition $n \geq 3t + 1$ (see the proof of $\OciorRBC$ in~\cite{ChenOciorCOOL:24}).

 \begin{definition}[{\bf Honest-Majority Distributed Multicast ($\HMDM$)} \cite{Chen:2020arxiv, ChenDISC:21, ChenOciorCOOL:24}] \label{def:DM}
In the \emph{distributed multicast} ($\DM$) problem, a subset of nodes acts as senders that multicast messages to all $n$ nodes, where up to $t$ nodes may be dishonest. Each sender node has its own input message.  
A protocol is called a $\DM$ protocol if it satisfies the following property:
\begin{itemize}
\item {\bf Validity:} If all honest senders input the same message $\wv$, then every honest node eventually outputs $\wv$.
\end{itemize}
The $\DM$ problem is referred to as an \emph{honest-majority distributed multicast} ($\HMDM$) if at least $t + 1$ of the senders are honest. The $\HMDM$ abstraction was used as a building block in the $\COOL$ protocol~\cite{Chen:2020arxiv, ChenDISC:21, ChenOciorCOOL:24}.
\end{definition}

\begin{definition}[{\bf Strongly-Honest-Majority Distributed Multicast ($\SHMDM$)} \cite{ChenOciorMVBA:24}]  \label{def:SHMDM}
A distributed multicast problem is called a {\em strongly-honest-majority distributed multicast}  if the set of senders is fixed and denoted by $\Sc \subseteq [n]$, where $|\Sc| \ge 3t + 1$. 
A $\SHMDM$ protocol satisfies the following property:
\begin{itemize}
    \item {\bf Validity:} If all honest senders in $\Sc$ input the same message $\wv$,  then every honest node eventually outputs $\wv$.
\end{itemize}
\end{definition}

\begin{definition}[{\bf Unique Agreement ($\BUA$)} \cite{Chen:2020arxiv, ChenDISC:21, ChenOciorCOOL:24}] \label{def:BUA}
In a $\BUA$ protocol, each node begins with an initial input value and aims to produce an output of the form $(\MVBAInputMsg, \successindicator, \Vr)$, where $\successindicator \in \{0,1\}$ denotes a success indicator and $\Vr \in \{0,1\}$ represents a vote. 
A $\BUA$ protocol must satisfy the following properties:
\begin{itemize}
    \item {\bf Unique Agreement:} If two honest nodes output $(\wv', 1, *)$ and $(\wv'', 1, *)$, respectively, then $\wv' = \wv''$.
    \item {\bf Majority Unique Agreement:} If an honest node outputs $(*, *, 1)$, then at least $t + 1$ honest nodes eventually output $(\wv, 1, *)$ for the same $\wv$.
    \item {\bf Validity:} If all honest nodes begin with the same input value $\wv$, then all honest nodes eventually output $(\wv, 1, 1)$.
\end{itemize}
\end{definition}

\vspace{3pt}
\emph{Notations}: In an asynchronous setting, when measuring communication rounds, we use \emph{asynchronous rounds}, where each round need not be synchronized across the distributed nodes. 
Let $:=$ denote ``is defined as.'' 
Let $\defaultvalue$ denote a default or empty value.  
Let $[b] := \{1, 2, \dotsc, b\}$ and $[a,b] := \{a, a+1, a+2, \dotsc, b\}$ for integers $a$ and $b$. 
We use $f(x) = O(g(x))$ to denote that $\limsup_{x \to \infty} \frac{|f(x)|}{g(x)} < \infty$. 
We use $f(x) = \Omega(g(x))$ to denote that $\liminf_{x \to \infty} \frac{f(x)}{g(x)} > 0$. 
We use $f(x) = \tilde{O}(g(x))$ to mean that $\limsup_{x \to \infty} \frac{|f(x)|}{(\log x)^a g(x)} < \infty$ for some constant $a \ge 0$.

\begin{algorithm}
\caption{$\OciorACOOL$ protocol with identifier $ \IDCOOL $.  Code is shown for    Node~$\thisnodeindex$ for $ \thisnodeindex \in [n]$. }  \label{algm:OciorACOOL}
\begin{algorithmic}[1]

\footnotesize

  \Statex   \emph{//   ** We use   $\RyNew_i^{[1]}, \RyNew_i^{[2]}, \SsNew_1^{[1]}, \SsNew_0^{[1]}, \SsNew_1^{[2]}, \SsNew_0^{[2]}, \VrNew_i, \MeNew^{(i)}, \{(\ysymbolNew_i^{(j)}, \ysymbolNew_j^{(j)})\}_j$   to denote the corresponding values delivered from  $\COOLBUA[1]$ **}

  \Statex   \emph{//   ** We use  $\RySecond_i^{[1]}, \RySecond_i^{[2]}, \SsSecond_1^{[1]}, \SsSecond_0^{[1]}, \SsSecond_1^{[2]}, \SsSecond_0^{[2]}, \VrSecond_i, \MeSecond^{(i)}, \{(\ysymbolSecond_i^{(j)}, \ysymbolSecond_j^{(j)})\}_j$   to denote the corresponding values delivered from  $\COOLBUA[2]$ **}

\State  Initially set   $\kencode \gets      \networkfaultsizet / 3; \MeSecond_i\gets \defaultvalue; \Me^{(i)}\gets \defaultvalue;  \OECSI=0; \OECSIFinal\gets 0;  \OECCorrectSymbolSetNew\gets \{\}; \OECCorrectSymbolSet\gets \{\};   \ySymbolSet\gets \{\};  \Phthreeindicator\gets 0; \ystrongmajority_i \gets \defaultvalue$    	 
\State  $\wait$ until the delivery of   $\RyNew_i^{[1]}$, $\RyNew_i^{[2]}$, $\SsNew_1^{[1]}$, $\SsNew_0^{[1]}$, $\SsNew_1^{[2]}$, $\SsNew_0^{[2]}$     from $\COOLBUA[1]$ and    $\RySecond_i^{[1]}$, $\RySecond_i^{[2]}$, $\SsSecond_1^{[1]}$, $\SsSecond_0^{[1]}$, $\SsSecond_1^{[2]}$, $\SsSecond_0^{[2]}$     from $\COOLBUA[2]$

 \vspace{0pt}

  \Statex   \emph{//   ********************  $\COOLBUA[1]$********************} 
	
\State {\bf upon} receiving a non-empty  message  input $\Me_{i}$ {\bf do}:

\Indent  

 \State    $\Pass$ $\Me_{i}$ into  $\COOLBUA[1]$    as an input value
 
\EndIndent 

 \vspace{0pt}
 
  \Statex   \emph{//   ******************** Set the   Input  Value $\MeSecond_i$ for  $\COOLBUA[2]$********************}

\State {\bf upon} receiving   $\ltuple\NEWSYMBOL, \IDCOOL, \ystrongmajority_j \rtuple$  from  Node~$j$ for the first time,   and $j\notin \OECCorrectSymbolSetNew$   {\bf do}:      \label{line:BUA2InputYiMajoritySi0Cond} 	
\Indent  
	\State $\OECCorrectSymbolSetNew[j] \gets \ystrongmajority_j$         \label{line:BUA2InputYiMajorityyj}
 	\If  { $|\OECCorrectSymbolSetNew|\geq  \kencode + \networkfaultsizet $ and $\OECSI=0$}     \label{line:BUA2OECbegin}   \quad\quad \quad \quad\quad\quad \quad   \quad\quad \quad \quad\quad\quad \quad   \quad\quad \quad \quad\quad\quad \quad    \emph{//   online error correcting  (OEC)  }  
			\State   $\MVBAOutputMsg  \gets \ECCDec(n,  \kencode , \OECCorrectSymbolSetNew)$	      \label{line:BUA1OECdec}	  
			\State  $[y_{1}, y_{2}, \cdots, y_{n}] \gets \ECCEnc (n,  \kencode, \MVBAOutputMsg)$ 
    			\If {at least $\kencode + \networkfaultsizet$ symbols in $[y_{1}, y_{2}, \cdots, y_{n}]$ match with  those in $\OECCorrectSymbolSetNew$}
 				  set $ \MeSecond_i \gets \MVBAOutputMsg$ and  $\OECSI \gets 1$     \label{line:BUA2OECend}	 
    			\EndIf    

	\EndIf   
		     
\EndIndent

 \State {\bf upon} $\COOLBUA[1]$ having  delivered  $\ltuple\SYMBOL, 1, (*, \ysymbolNew_j^{(j)})\rtuple$  and $\SsNew_1^{[1]}$  such that   $j\in \SsNew_1^{[1]}$, and $j\notin \OECCorrectSymbolSetNew$    {\bf do}:        \label{line:BUA2InputYiMajoritySi1Cond} 	
\Indent  
	\State $\OECCorrectSymbolSetNew[j] \gets \ysymbolNew_j^{(j)}$            \label{line:BUA2InputYiMajoritySi1} 	
	\State run the OEC steps as in Lines~\ref{line:BUA2OECbegin}-\ref{line:BUA2OECend}           \label{line:BUA2InputYiMajoritySi1OEC} 
\EndIndent

 \State {\bf upon} delivery of $\ltuple\SYMBOL, 1, (\ysymbolNew_i^{(j)},  *) \rtuple$   from   $\COOLBUA[1]$  {\bf do}:      \label{line:BUA2InputSymbolDeliverCond}         
  \Indent  
	\IfThenElse {$\ysymbolNew_i^{(j)} \in \ySymbolSet$}  {$\ySymbolSet[ \ysymbolNew_i^{(j)}] \gets \ySymbolSet[\ysymbolNew_i^{(j)}]\cup \{j\}$} {$\ySymbolSet[\ysymbolNew_i^{(j)}] \gets \{j\}$}         \label{line:BUA2InputSymbolDeliver}

\EndIndent

  \State {\bf upon}   $(|\ySymbolSet[y^{\star}] \cup \SsNew_0^{[2]}| \geq n-t)\AND (|\ySymbolSet[y^{\star}]| \geq n-2t)\AND(\RyNew_i^{[1]} \neq 1 )\AND (\ystrongmajority_i =\defaultvalue)$  for some $y^{\star}$    {\bf do}:       \label{line:BUA2InputYiMajorityCond} 	   \emph{//   $\RyNew_i^{[1]} ,  \SsNew_0^{[2]}$  updated from $\COOLBUA[1]$}          
    \Indent  
			\State set $\ystrongmajority_i\gets  y^{\star}$; and     
			  $\send$  $\ltuple\NEWSYMBOL, \IDCOOL, \ystrongmajority_i \rtuple$  to  all nodes	    \label{line:BUA2InputYiMajoritySend} 
\EndIndent

 \vspace{0pt}
 
  \Statex   \emph{//   ********************  $\COOLBUA[2]$ ********************} 
 
 \State {\bf upon} delivery of $\RyNew_i^{[2]} = 1$   from   $\COOLBUA[1]$,  and  $\COOLBUA[2]$ has no input message  yet {\bf do}:   \label{line:BUA2Input1Cond}   
\Indent  
	\State  set    $\MeSecond_i \gets  \Me_{i}$            \label{line:BUA2Input1Setw} 
 	\State $\Pass$  $\MeSecond_i$ into $\COOLBUA[2]$ as an input message       \label{line:BUA2Input1} 	
 
\EndIndent

\State {\bf upon} $\MeSecond_i \neq \defaultvalue$, and  $\COOLBUA[2]$ has no input message  yet  {\bf do}:        \label{line:BUA2Input2Cond}
\Indent  

	\State  $\Pass$  $\MeSecond_i$ into $\COOLBUA[2]$ as an input message          \label{line:BUA2Input2} 
\EndIndent

 \vspace{0pt}
 
  \Statex   \emph{//   ********************  $\ABBA$ ********************} 

\State {\bf upon} delivery of $[\MeSecond^{(i)}, \RySecond_{i}^{[2]},   \VrSecond_i ]$    from   $\COOLBUA[2]$, and  $\ABBA$ has no  input yet  {\bf do}:          \label{line:OciorACOOLABBAinputCond} 
\Indent  
	\State  $\Pass$  $\VrSecond_i$ into $\ABBA$ as an input message              \label{line:OciorACOOLABBAinput}     \quad \quad  \quad\quad \quad \quad \quad  \quad\quad \emph{//     an asynchronous  binary $\BA$  ($\ABBA$) protocol  } 
\EndIndent

\State {\bf upon} delivery of $[\MeNew^{(i)}, \RyNew_{i}^{[2]},   \VrNew_i =0 ]$ or   $\RyNew_i^{[2]} = 0$  from   $\COOLBUA[1]$, and  $\ABBA$ has no  input yet {\bf do}:  \label{line:OciorACOOLABBAinputBUA1Cond} 
 
\Indent  
	\State  $\Pass$  $0$ into $\ABBA$ as an input message              \label{line:OciorACOOLABBAinput1}     
\EndIndent

  \Statex   \emph{//   ******************** Binary  $\RBA$  ($\BRBA$)********************}

\State {\bf upon}  $\ABBA$  outputting $\ABBAOutput$   {\bf do}:  	      \label{line:OciorACOOLRBAbegin}  
\Indent  
		\State $\send$ $\ltuple\READY, \IDCOOL, \ABBAOutput \rtuple$ to  all nodes  	   \label{line:OciorACOOLRBAReadyABBA}  
\EndIndent

\State {\bf upon} receiving   $\networkfaultsizet+1$  $\ltuple \READY, \IDCOOL, \Vr  \rtuple$ messages  from different   nodes for the same $\Vr$, and $\ltuple\READY, \IDCOOL, * \rtuple$  not yet sent {\bf do}:        \label{line:OciorACOOLRelialbeBegin}   
\Indent  
		\State $\send$ $\ltuple\READY, \IDCOOL, \Vr \rtuple$ to  all nodes	    \label{line:OciorACOOLRBAsendReady}   
\EndIndent

\State {\bf upon} receiving   $2t+1$  $\ltuple \READY, \IDCOOL, \Vr  \rtuple$ messages  from different nodes   for the same $\Vr$ {\bf do}:     \label{line:OciorACOOLVrOutputCond}   
\Indent  
 	\State  set $\VrOutput\gets \Vr$             \label{line:OciorACOOLVrOutput}       
     \IfThenElse {$\VrOutput =0$}  {$\Output$   $\Me^{(i)}=\defaultvalue$ and $\terminate$ } {set $\Phthreeindicator\gets 1$ }     \label{line:OciorACOOLph3}     
     
\EndIndent

 \vspace{0pt}
  
  \Statex   \emph{//   ********************  $\COOLHMDM[2]$ ********************}

\State {\bf upon} $\Phthreeindicator= 1$ {\bf do}:       \label{line:OciorACOOLph3trigger}   \quad \quad\quad\quad\quad\quad  \quad \quad\quad \quad \quad\quad\quad\quad\quad  \quad \quad\quad\quad\quad\quad\quad \quad      \emph{//     only after executing Line~\ref{line:OciorACOOLph3}}
\Indent  
\If { $\COOLBUA[2]$ has delivered $[\MeSecond^{(i)}, \RySecond_{i}^{[2]},   \VrSecond_i ]$ with $\RySecond_i^{[2]} = 1$}          \label{line:OciorACOOLph3SIPhtwo}
	\State $\Output$   $\MeSecond^{(i)}$ and $\terminate$     \label{line:OciorACOOLph3SIPhtwoOutput}
		 
\Else
 
			\State  $\wait$ until $\COOLBUA[2]$ delivering at least  $\networkfaultsizet +1$ $\ltuple\SYMBOL, 2, (\ysymbolSecond_i^{(j)},  *) \rtuple$, $\forall j \in   \SsSecond_1^{[2]}$, for the same    $\ysymbolSecond_i^{(j)} =y^{\star}$,  for some   $y^{\star}$   \label{line:OciorACOOLph3MajorityRuleCond}
			\State  $\ysymbolSecond_i^{(i)} \gets y^{\star}$   \label{line:OciorACOOLph3MajorityRule}   \quad\quad\quad\quad \quad\quad\quad\quad\quad\quad\quad\quad\quad\quad\quad\quad   \quad\quad\quad\quad\quad\quad\quad\quad \quad\quad \emph{// update coded symbol based on  majority rule}  
	\State   $\send$   $\ltuple \CORRECTSYMBOL, \IDCOOL, \ysymbolSecond_i^{(i)}  \rtuple$  to  all nodes      \label{line:OciorACOOLph3SendCorrectSymbols} 
			\State  $\wait$ until  $\OECSIFinal=1$ 	
			\State $\Output$   $\Me^{(i)}$ and $\terminate$     \label{line:OciorACOOLph3Output2}	
	
\EndIf

\EndIndent

\State {\bf upon} receiving   $\ltuple\CORRECTSYMBOL, \IDCOOL, \ysymbolSecond_j^{(j)} \rtuple$  from  Node~$j$ for the first time,   $j\notin \OECCorrectSymbolSet$, and $\OECSIFinal=0$   {\bf do}:  \label{line:Ph3OECCond}
\Indent  
	\State $\OECCorrectSymbolSet[j] \gets \ysymbolSecond_j^{(j)}$    
 	\If  { $|\OECCorrectSymbolSet|\geq  \kencode + \networkfaultsizet  $}     \label{line:OECbegin}    \quad    \quad\quad \quad \quad\quad\quad \quad    \quad\quad \quad \quad\quad\quad \quad   \quad\quad \quad \quad\quad\quad \quad   \quad\quad \quad \quad\quad\quad \quad    \emph{//   online error correcting   }  
			\State   $\MVBAOutputMsg  \gets \ECCDec(n,  \kencode , \OECCorrectSymbolSet)$	     
			\State  $[y_{1}, y_{2}, \cdots, y_{n}] \gets \ECCEnc (n,  \kencode, \MVBAOutputMsg)$ 
    			\If {at least $\kencode + \networkfaultsizet$ symbols in $[y_{1}, y_{2}, \cdots, y_{n}]$ match with  those in $\OECCorrectSymbolSet$}
 				set  $ \Me^{(i)} \gets \MVBAOutputMsg$ and  $\OECSIFinal\gets 1$     \label{line:OECend}	 
    			\EndIf    

	\EndIf   
		     
\EndIndent

\State {\bf upon} $\COOLBUA[2]$ having  delivered  $\ltuple\SYMBOL, 2, (*, \ysymbolSecond_j^{(j)})\rtuple$  and $\SsSecond_1^{[2]}$  such that   $j\in \SsSecond_1^{[2]}$, and $j\notin \OECCorrectSymbolSet$, and  $\OECSIFinal=0$   {\bf do}:    \label{line:ACOOLPh3OECSecCond}
\Indent  
	\State $\OECCorrectSymbolSet[j] \gets \ysymbolSecond_j^{(j)}$     
	\State run the OEC steps as in Lines~\ref{line:OECbegin}-\ref{line:OECend}      \label{line:Ph3OECAllEnd}
\EndIndent

\end{algorithmic}
\end{algorithm}

\begin{algorithm}
\caption{$\COOLBUA$ protocol in asynchronous setting, with identifier $ \IDCOOL$.  Code is shown for    Node~$\thisnodeindex$. }  \label{algm:COOLBUA}
\begin{algorithmic}[1]
\vspace{5pt}    
 
\footnotesize
 
 \Statex   \emph{//   **  This $\COOLBUA$ protocol continuously updates and delivers the values $\Ry_i^{[1]}, \Ry_i^{[2]}, \Ss_1^{[1]}, \Ss_0^{[1]}, \Ss_1^{[2]}, \Ss_0^{[2]}, \Vr_i, \Me^{(i)}, \{(y_i^{(j)}, y_j^{(j)})\}_j$ to the main protocol that invokes it. It   terminates when the main protocol terminates. **}

 \Statex
 
\State  Initially set   $\kencode \gets      \networkfaultsizet / 3;    \Lkset_0\gets \{\}; \Lkset_1\gets \{\}; \Ss_0^{[1]}\gets \{\};\Ss_1^{[1]}\gets \{\};\Ss_0^{[2]}\gets \{\};\Ss_1^{[2]}\gets \{\};    \Me^{(i)}\gets \defaultvalue;   \ECCEncindicator \gets 0;     \Ry_i^{[1]} \gets \defaultvalue;    \Ry_i^{[2]} \gets \defaultvalue;   \Vr_i \gets \defaultvalue $  	\quad \quad      \emph{//  $\defaultvalue$ is a default value or an empty value  }    
 
  	\State  $\deliver$   $\Ry_i^{[1]}$, $\Ry_i^{[2]}$, $\Ss_1^{[1]}$, $\Ss_0^{[1]}$, $\Ss_1^{[2]}$, $\Ss_0^{[2]}$

  \Statex  
  
\Statex {\bf \emph{Phase~1}}       

\State {\bf upon} receiving a non-empty  message  input $\Me_{i}$ {\bf do}:

\Indent

 \State    $\Me^{(i)}\gets \Me_{i}$;   \   $[y_1^{(i)}, y_2^{(i)}, \cdots, y_{n}^{(i)}]\gets \ECCEnc(n, \kencode, \Me_{i})$     \label{line:ACOOLECCEnc}
 \State   $\send$  $\ltuple   \SYMBOL, \IDCOOL,  (y_j^{(i)}, y_i^{(i)}) \rtuple$ to Node~$j$,  $\forall j \in  [\networksizen]$  		 \label{line:ACOOLECCEncindicator}
 \State     $\ECCEncindicator\gets 1$   
   
\EndIndent

\State {\bf upon} receiving   $\ltuple\SYMBOL, \IDCOOL, (y_i^{(j)}, y_j^{(j)}) \rtuple$  from Node~$j$ for the first time  {\bf do}:  
\Indent  
	\State  $\wait$ until  $\ECCEncindicator =1$ 
	
    \IfThenElse {$ (y_i^{(j)}, y_j^{(j)}) = (y_i^{(i)}, y_j^{(i)})$}  {$\Lkset_1\gets \Lkset_1\cup \{j\}$ } {$\Lkset_0\gets \Lkset_0\cup \{j\}$}     \label{line:ACOOLph1MatchCond}

 	\State  $\deliver$ $\ltuple\SYMBOL, \IDCOOL,  (y_i^{(j)}, y_j^{(j)}) \rtuple$   
\EndIndent

\State {\bf upon}  $(|\Lkset_1| \geq    n- \networkfaultsizet)\AND(\Ry_i^{[1]} =\defaultvalue)$ {\bf do}:          \label{line:COOLBUASISI1oneCond} 
\Indent  
           \State	set  $\Ry_i^{[1]} \gets 1$;   $\send$ $\ltuple\SIone, \IDCOOL, \Ry_i^{[1]}\rtuple$ to all nodes; and  $\deliver$ $\Ry_i^{[1]}$    \label{line:COOLBUASISI1oneSend}
 
\EndIndent

\State {\bf upon}  $(|\Lkset_0|  \geq    \networkfaultsizet +1)\AND (\Ry_i^{[1]} =\defaultvalue)$ {\bf do}:          \label{line:COOLBUASI1zeroCond} 
\Indent  
           \State	set $\Ry_i^{[1]} \gets 0$;   $\send$ $\ltuple\SIone, \IDCOOL, \Ry_i^{[1]}\rtuple$ to all nodes; and  $\deliver$ $\Ry_i^{[1]}$      \label{line:COOLBUASI1zero} 
\EndIndent

\State {\bf upon} receiving   $\ltuple\SIone, \IDCOOL, \Ry_j^{[1]}\rtuple$  from  Node~$j$ for the first time   {\bf do}:    \label{line:ACOOLph1SS01Cond}  
\Indent  

     	\IfThenElse {$\Ry_j^{[1]} =1$}  {$\Ss_1^{[1]}\gets \Ss_1^{[1]}\cup \{j\}$} {$\Ss_0^{[1]} \gets \Ss_0^{[1]}\cup \{j\}$}     \label{line:ACOOLph1SS0NotM}          
	\State  $\deliver$ $\Ss_1^{[1]}$ and $\Ss_0^{[1]}$      \label{line:UAS1S0p1deli}    
\EndIndent

  \Statex

\Statex  {\bf \emph{Phase~2}}    

 \State {\bf upon}  $((\Ry_i^{[1]} = 0)\OR(|\Ss_0^{[1]}\cup \Lkset_0| \geq    \networkfaultsizet +1))\AND( \Ry_i^{[2]} = \defaultvalue)$  {\bf do}:  \label{line:ACOOLph2CSI0Cond}   
\Indent    

	\State	set $\Ry_i^{[2]} \gets 0$;       $\send$ $\ltuple\SItwo, \IDCOOL, \Ry_i^{[2]}\rtuple$ to all nodes; and  $\deliver$ $\Ry_i^{[2]}$      \label{line:ACOOLph2CSI0}

\EndIndent

 \State {\bf upon}  $(\Ry_i^{[1]} = 1) \AND(|\Ss_1^{[1]}\cap \Lkset_1|\geq  n- \networkfaultsizet)\AND( \Ry_i^{[2]} = \defaultvalue)$ {\bf do}:  \label{line:ACOOLph2ACond} 
\Indent   

           \State	set  $\Ry_i^{[2]} \gets 1$;  $\send$ $\ltuple \SItwo, \IDCOOL,\Ry_i^{[2]}\rtuple$ to all nodes; and $\deliver$ $\Ry_i^{[2]}$   \label{line:ACOOLph2ASI1} 
	
\EndIndent

 \State {\bf upon} receiving   $\ltuple\SItwo, \IDCOOL, \Ry_j^{[2]}\rtuple$  from  Node~$j$ for the first time  {\bf do}:     \label{line:ACOOLph2SS01Cond}     
\Indent  
     	\IfThenElse {$\Ry_j^{[2]} =1$}  {$\Ss_1^{[2]}\gets \Ss_1^{[2]}\cup \{j\}$} {$\Ss_0^{[2]} \gets \Ss_0^{[2]}\cup \{j\}$}      \label{line:ACOOLph2SS0NM}

 	\State  $\deliver$ $\Ss_1^{[2]}$ and $\Ss_0^{[2]}$   \label{line:UAS1S0p2deli}    
 	 
\EndIndent       
 
  \Statex

\State {\bf upon}   $|\Ss_{1}^{[2]}|\geq  n- \networkfaultsizet$ {\bf do}:      \label{line:OciorACOOLReadyOneCondition}    
\Indent  
		\State   $\Vr_i \gets 1$; and     $\deliver$  $[\Me^{(i)}, \Ry_{i}^{[2]},   \Vr_i ]$  	 \label{line:OciorACOOLReadyOne}   

\EndIndent

\State {\bf upon}   $|\Ss_{0}^{[2]}|\geq  t +1$ {\bf do}:      \label{line:OciorACOOLReadyZeroCondition}   
\Indent  
		\State   $\Vr_i \gets 0$;  and    $\deliver$  $[\Me^{(i)}, \Ry_{i}^{[2]},   \Vr_i ]$  	     \label{line:OciorACOOLReadyZero}   
\EndIndent

\end{algorithmic}
\end{algorithm}

\section{$\OciorACOOL$}    \label{sec:OciorACOOL}

The proposed $\OciorACOOL$ is an error-free asynchronous  Byzantine agreement  protocol.  
$\OciorACOOL$ does not rely on any cryptographic assumptions such as signatures or hashing, except for  a single invocation of an asynchronous  binary $\BA$ protocol. 
 $\OciorACOOL$ is an adaptive variant of $\COOL$.

\subsection{Revisiting $\COOL$ for the Synchronous Setting} \label{sec:OverviewOciorCOOL}

Before presenting the proposed $\OciorACOOL$ protocol for the asynchronous setting, we first revisit the original $\COOL$ (or $\OciorCOOL$) protocol designed for the synchronous setting. 
Note that the original $\COOL$ protocol consists of the following   components: the $\COOLBUA$, binary $\BA$,  and the $\COOLHMDM$ algorithms, as shown in Fig.~\ref{fig:OciorCOOLProtocol}. 

\subsubsection{$\COOLBUA$} 
$\COOLBUA$ is a $\BUA$ algorithm in which each node~$i$ inputs an initial value $\MVBAInputMsg_i$ and aims to produce an output $(\MVBAInputMsg, \successindicator, \Vr)$, where $\MVBAInputMsg$ is the updated message, $\successindicator \in \{0,1\}$ is a success indicator, and $\Vr \in \{0,1\}$ is a vote.  
$\COOLBUA$ guarantees three properties: \emph{Unique Agreement}, \emph{Majority Unique Agreement}, and \emph{Validity} (see Definition~\ref{def:BUA}). 

\subsubsection{Binary $\BA$} 
After obtaining outputs from $\COOLBUA$, each node~$i$ passes its vote $\Vr$ to the binary $\BA$ consensus as input.  
The binary $\BA$ consensus ensures that all honest nodes reach the same decision on whether to terminate at the end of the binary $\BA$ or proceed to the $\COOLHMDM$ phase. 

\subsubsection{$\COOLHMDM$} 
By the Consistency and Validity properties of the binary $\BA$ consensus, if any honest node enters the $\COOLHMDM$ phase, it is guaranteed that at least one honest node has voted $\Vr = 1$ in the binary $\BA$ consensus.  
In this case, by the Unique Agreement and Majority Unique Agreement properties of $\BUA$, at least $t + 1$ honest nodes must have output $(\wv, 1, *)$ from $\COOLBUA$ for the same $\wv$.  
This condition satisfies the $\HMDM$ requirement in distributed multicast, i.e., at least $t + 1$ senders are honest.  

$\HMDM$ guarantees the \emph{Validity} property: if all honest senders (with at least $t + 1$ honest senders) input the same message $\wv$, then every honest node eventually outputs $\wv$ in $\HMDM$.  
$\COOLHMDM$ is an $\HMDM$ algorithm that ensures that the encoded symbols from honest nodes can be calibrated using a majority rule, so that the calibrated symbols are encoded from the same    $\wv$.  
Consequently, all honest nodes eventually decode and output the same final message.

\subsection{The Challenge of $\COOL$ in the Asynchronous Setting} \label{sec:issue}

It can be verified that the $\COOL$ protocol satisfies the \emph{Consistency} and \emph{Validity} properties in the asynchronous setting. 
However, the main challenge lies in ensuring \emph{termination} for the $\COOL$ protocol under asynchrony. 
To illustrate this challenge, we consider the following example. 

Let $\Me_i$ denote the initial message of node~$i$, and let $\Fc$ denote the index set of all dishonest nodes.  
We define the index groups of honest nodes as 
\[
\Ac_{\Lin} \defeqnew \{ i \in [n] \setminus \Fc \mid \Me_i = \Meg_{\Lin} \},
\]
for $\Lin \in [\gn]$, where $\Meg_{1}, \Meg_{2}, \dots, \Meg_{\gn}$ are distinct non-empty $\ell$-bit values and $\gn$ is an integer.   
Let us consider the following example:
\[
\text{Example:} \quad n = 3t + 1, \quad |\Ac_1| = n- t -|\Fc| , \quad |\Ac_2| = t, \quad |\Fc| = t.
\]

In this asynchronous setting, the honest nodes in $\Ac_1$ may be unable to produce outputs $(\MVBAInputMsg, \successindicator, \Vr)$ in $\COOLBUA$ (see Algorithm~\ref{algm:COOLBUA}).  
Intuitively, for this example, as shown in Algorithm~\ref{algm:COOLBUA}, each honest node~$i$ in $\Ac_2$ eventually sets its success indicators to $\Ry_i^{[2]} = \Ry_i^{[1]} = 0$ 
(see Lines~\ref{line:ACOOLph1MatchCond}, \ref{line:COOLBUASI1zero}, and \ref{line:ACOOLph2CSI0}) 
because of the mismatched initial messages between the honest nodes in $\Ac_1$ and $\Ac_2$ 
(see Line~\ref{line:ACOOLph1MatchCond}). 
However, each honest node~$i$ in $\Ac_1$ may be unable to update the values of $\Ry_i^{[2]}$ and $\Ry_i^{[1]}$ if the dishonest nodes do not send messages to the honest nodes in $\Ac_1$ (see Algorithm~\ref{algm:COOLBUA}). 

In the following, we introduce the $\OciorACOOL$ protocol, an adaptive variant of $\COOL$, which resolves the termination challenge in the asynchronous setting while preserving the \emph{Consistency} and \emph{Validity} properties.

\subsection{Overview of $\OciorACOOL$ for asynchronous stetting}    \label{sec:OverviewOciorACOOL}

The proposed $\OciorACOOL$ protocol is described in Algorithm~\ref{algm:OciorACOOL} and supported by Algorithm~\ref{algm:COOLBUA}. 
We here provide an overview of the proposed $\OciorACOOL$, which mainly consists of the $\COOLBUA[1]$, $\COOLBUA[2]$, asynchronous  binary $\BA$  ($\ABBA$),   binary $\RBA$ ($\BRBA$),  and $\COOLHMDM[2]$ algorithms. 
Note that the original $\COOL$ protocol comprises the $\COOLBUA$, binary $\BA$ ($\BBA$),  and $\COOLHMDM$ algorithms. 

As shown in the previous subsection, in the asynchronous setting, some honest nodes may be unable to produce outputs $(\MVBAInputMsg, \successindicator, \Vr)$ in $\COOLBUA[1]$ (see Algorithm~\ref{algm:COOLBUA}).  
The design of $\OciorACOOL$ guarantees that all honest nodes eventually \emph{terminate}, while preserving the \emph{Consistency} and \emph{Validity} properties. 
Before describing the $\OciorACOOL$ protocol, we first introduce several definitions. 

In $\COOLBUA$, the values of $\Ry_i^{[1]}$ and $\Ry_i^{[2]}$ are initially set to $\Ry_i^{[1]} = \defaultvalue$ and $\Ry_i^{[2]} = \defaultvalue$, respectively, where $\defaultvalue$ denotes a default or empty value. 
Let
\[
\AsetNOzero_{\Lin}^{[p]} \defeqnew \{\, i \in [n]\setminus \Fc \mid \Me_i = \Meg_{\Lin},\ \Ry_i^{[p]} \neq 0 \,\}, 
\quad \Lin \in [\gnNozero^{[p]}], \quad p \in [2],
\]
for some non-negative integers $\gn$, $\gnNozero^{[1]}$, and $\gnNozero^{[2]}$ satisfying $\gnNozero^{[2]} \leq \gnNozero^{[1]} \leq \gn$. 
In the definition of $\AsetNOzero_{\Lin}^{[p]}$, we consider the final value of $\Ry_i^{[p]}$ if it has been updated. 
Intuitively, $\AsetNOzero_{\Lin}^{[p]}$ represents the indices of honest nodes that eventually set the success indicator $\Ry_i^{[p]} = 1$ or never update $\Ry_i^{[p]}$. 

We prove in Lemma~\ref{lm:OciorACOOBUAUniqueSet} that if all honest nodes eventually input their initial messages and keep running the $\COOLBUA$ protocol, then it holds that 
\[
\gnNozero^{[2]} \leq 1,
\]
i.e., every honest node $i$ eventually sets $\Ry_i^{[2]} = 0$ for all $i \in \Ac_{\Lin}$ and all $\Lin \in [2, \gn]$.  
In the following, we consider two possible cases for $\COOLBUA[1]$:
\begin{align}
\text{Case~I:} \quad &|\AsetNOzero_{1}^{[2]}| \geq n - |\Fc| - t,  \non \\
\text{Case~II:} \quad &|\AsetNOzero_{1}^{[2]}| < n - |\Fc| - t.  \non
\end{align}
The $\OciorACOOL$ protocol guarantees that, by combining the following algorithms, for each of the above two cases, all honest nodes eventually \emph{terminate}, while preserving the \emph{Consistency} and \emph{Validity} properties.

 For notational convenience,  we use   $\RyNew_i^{[1]}, \RyNew_i^{[2]}, \SsNew_1^{[1]}, \SsNew_0^{[1]}, \SsNew_1^{[2]}, \SsNew_0^{[2]}, \VrNew_i, \MeNew^{(i)}, \{(\ysymbolNew_i^{(j)}, \ysymbolNew_j^{(j)})\}_j$   to denote the corresponding values delivered from  $\COOLBUA[1]$.  Similarly, we use  $\RySecond_i^{[1]}, \RySecond_i^{[2]}, \SsSecond_1^{[1]}, \SsSecond_0^{[1]}, \SsSecond_1^{[2]}, \SsSecond_0^{[2]}, \VrSecond_i, \MeSecond^{(i)}$, $\{(\ysymbolSecond_i^{(j)}, \ysymbolSecond_j^{(j)})\}_j$   to denote the corresponding values delivered from  $\COOLBUA[2]$.    

 Here, $\RyNew_i^{[1]}$ and $\RyNew_i^{[2]}$  (and $\RySecond_i^{[1]}$ and $\RySecond_i^{[2]}$) denote the success indicators of Node~$i$ at Phase~1 and Phase~2, respectively (see Lines~\ref{line:COOLBUASISI1oneSend}, \ref{line:COOLBUASI1zero}, \ref{line:ACOOLph2CSI0}, and \ref{line:ACOOLph2ASI1} of Algorithm~\ref{algm:COOLBUA}). 
$\SsNew_b^{[p]}$ (and $\SsSecond_b^{[p]}$) denotes the index set containing the nodes that send $\RyNew_i^{[p]} = b$, for $b \in \{0,1\}$ and $p \in \{1,2\}$ (see Lines~\ref{line:UAS1S0p1deli} and \ref{line:UAS1S0p2deli} of Algorithm~\ref{algm:COOLBUA}).  
 $\VrNew_i$ (and $\VrSecond_i$)  denotes the binary vote value. $(\ysymbolNew_i^{(j)}, \ysymbolNew_j^{(j)})$ (and $(\ysymbolSecond_i^{(j)}, \ysymbolSecond_j^{(j)})$) denotes the pair consisting of   the $i$-th and the $j$-th coded symbols sent from  Node~$j$.

 \subsubsection{$\COOLBUA[1]$} 
$\COOLBUA[1]$ is a $\BUA$ algorithm, presented in Algorithm~\ref{algm:COOLBUA}, in which each node~$i$ inputs an initial value $\MVBAInputMsg_i$ and aims to produce an output $(\MVBAInputMsg, \successindicator, \Vr)$, where $\MVBAInputMsg$ is the updated message, $\successindicator \in \{0,1\}$ is a success indicator, and $\Vr \in \{0,1\}$ is a vote.  
$\COOLBUA[1]$ guarantees three properties: \emph{Unique Agreement}, \emph{Majority Unique Agreement}, and \emph{Validity} (see Definition~\ref{def:BUA}). 

As mentioned earlier, some honest nodes may be unable to produce outputs $(\MVBAInputMsg, \successindicator, \Vr)$ in $\COOLBUA[1]$.  
However, $\COOLBUA[1]$ may still update and deliver the following values and sets to the main protocol that invokes it:  
$\RyNew_i^{[1]}, \RyNew_i^{[2]}, \SsNew_1^{[1]}, \SsNew_0^{[1]}, \SsNew_1^{[2]}, \SsNew_0^{[2]}$, and $\{(\ysymbolNew_i^{(j)}, \ysymbolNew_j^{(j)})\}_j$.

  \subsubsection{Set the   Input  Value $\MeSecond_i$ for  $\COOLBUA[2]$} 

The goal of this phase is to set the input message for $\COOLBUA[2]$ by executing the steps in Lines~\ref{line:BUA2InputYiMajorityyj}-\ref{line:BUA2InputYiMajoritySend} of Algorithm~\ref{algm:OciorACOOL}.  
This phase ensures that, in Case~I, if all honest nodes eventually input their initial messages to $\OciorACOOL$  and keep running the $\OciorACOOL$   protocol,  then all honest nodes eventually input the same value $\MeSecond_i = \Meg_1$ into $\COOLBUA[2]$ as the input messages (see Lemma~\ref{lm:OciorACOOLCaseISameInputBUA2}).  Here,  $\Meg_{1}$ is the initial message of all honest nodes within $\AsetNOzero_{1}^{[2]}$ for $\COOLBUA[1]$.

\subsubsection{$\COOLBUA[2]$} 
$\COOLBUA[2]$ is also a $\BUA$ algorithm,  which guarantees three properties: \emph{Unique Agreement}, \emph{Majority Unique Agreement}, and \emph{Validity}. 
 In Case~I,  if all honest nodes  keep running the  $\OciorACOOL$    protocol,  and  given the conclusion that  in this case every honest node~$i$ eventually inputs the same  message $\MeSecond_i=\Meg_{1}$ to $\COOLBUA[2]$ (see Lemma~\ref{lm:OciorACOOLCaseISameInputBUA2}), then  every honest node~$i$ eventually sets $\RySecond_i^{[1]} =1$ (see Lines~\ref{line:COOLBUASISI1oneCond} and \ref{line:COOLBUASISI1oneSend} of Algorithm~\ref{algm:COOLBUA}),
every honest node~$i$ eventually sets $\RySecond_i^{[2]} =1$ (see Lines~\ref{line:ACOOLph2ACond} and \ref{line:ACOOLph2ASI1} of Algorithm~\ref{algm:COOLBUA}),
 and every honest node~$i$ eventually sets  $\VrSecond_i = 1$ (see Lines~\ref{line:OciorACOOLReadyOneCondition} and \ref{line:OciorACOOLReadyOne} of Algorithm~\ref{algm:COOLBUA}).
  In this case,  every honest node~$i$ eventually   delivers $[\MeSecond^{(i)}, \RySecond_{i}^{[2]}, \VrSecond_i = 1]$ from $\COOLBUA[2]$ (see Line~\ref{line:OciorACOOLReadyOne} of Algorithm~\ref{algm:COOLBUA}).

\subsubsection{$\ABBA$ and   $\BRBA$}  
$\ABBA$ executes the steps in Lines~\ref{line:OciorACOOLABBAinputCond}-\ref{line:OciorACOOLABBAinput1} of Algorithm~\ref{algm:OciorACOOL}, 
while $\BRBA$ executes the steps in Lines~\ref{line:OciorACOOLRBAbegin}-\ref{line:OciorACOOLph3}.  
 
In Case~I, every honest node~$i$ eventually inputs the value~$1$ into $\ABBA$, if node~$i$ has not already provided an input to $\ABBA$ (see Lines~\ref{line:OciorACOOLABBAinputCond} and~\ref{line:OciorACOOLABBAinput} of Algorithm~\ref{algm:OciorACOOL}). 
Hence,  in Case~I, $\ABBA$ eventually outputs a value $\ABBAOutput \in \{1, 0\}$ at each node, due to the \emph{Termination} and \emph{Consistency} properties of $\ABBA$. 
Then, every honest node eventually sends $\ltuple \READY, \IDCOOL, \ABBAOutput \rtuple$ to all nodes in Line~\ref{line:OciorACOOLRBAReadyABBA} of Algorithm~\ref{algm:OciorACOOL}, and every honest node eventually sets $\VrOutput = \Vr^\star$ in Line~\ref{line:OciorACOOLVrOutput} of Algorithm~\ref{algm:OciorACOOL}.

In Case~II, the condition $|\Ss_{0}^{[2]}| \geq t + 1$ in Line~\ref{line:OciorACOOLReadyZeroCondition} of Algorithm~\ref{algm:COOLBUA} is eventually satisfied at all honest nodes, and each honest node~$i$ eventually sets $\VrNew_i = 0$ and delivers $[\MeNew^{(i)}, \RyNew_{i}^{[2]}, \VrNew_i = 0]$ from $\COOLBUA[1]$. 
Consequently, in Case~II each honest node~$i$ eventually passes $0$ into $\ABBA$ as its input message, if node~$i$ has not already provided an input to $\ABBA$ (see Lines~\ref{line:OciorACOOLABBAinputBUA1Cond} and~\ref{line:OciorACOOLABBAinput1} of Algorithm~\ref{algm:OciorACOOL}). 
In this case, every honest node eventually provides an input to $\ABBA$, and hence $\ABBA$ eventually outputs a value $\ABBAOutput \in \{1, 0\}$ at each node, due to the \emph{Termination} property of $\ABBA$. 
Then, every honest node eventually sends $\ltuple \READY, \IDCOOL, \ABBAOutput \rtuple$ to all nodes in Line~\ref{line:OciorACOOLRBAReadyABBA} of Algorithm~\ref{algm:OciorACOOL}, and every honest node eventually sets $\VrOutput = \Vr^\star$ in Line~\ref{line:OciorACOOLVrOutput} of Algorithm~\ref{algm:OciorACOOL}. 

If   $\Vr^\star = 0$, then all honest nodes eventually   output $\defaultvalue$ and terminate in Line~\ref{line:OciorACOOLph3} of Algorithm~\ref{algm:OciorACOOL}.
If $\Vr^\star = 1$,  then all honest nodes eventually  go to the $\COOLHMDM[2]$ phase.

 \subsubsection{$\COOLHMDM[2]$} 
 By the Consistency and Validity properties of the  $\ABBA$ consensus, if any honest node enters the $\COOLHMDM[2]$ phase, it is guaranteed that at least one honest node has voted $\Vr = 1$ in the  $\ABBA$ consensus.  
In this case, by the Unique Agreement and Majority Unique Agreement properties of $\COOLBUA[2]$, at least $t + 1$ honest nodes must have output $(\wv, 1, *)$ from $\COOLBUA[2]$ for the same $\wv$.  
This condition satisfies the $\HMDM$ requirement in distributed multicast, i.e., at least $t + 1$ senders are honest.  

$\HMDM$ guarantees the \emph{Validity} property: if all honest senders (with at least $t + 1$ honest senders) input the same message $\wv$, then every honest node eventually outputs $\wv$ in $\HMDM$.  
$\COOLHMDM[2]$ is an  $\HMDM$ algorithm, executed in Lines~\ref{line:OciorACOOLph3trigger}-\ref{line:Ph3OECAllEnd} of Algorithm~\ref{algm:OciorACOOL}.  
Therefore, all honest nodes eventually decode and output the same final message.

\subsection{Analysis of $\OciorACOOL$}    \label{sec:AnalysisOciorACOOL}

Our analysis closely follows that of $\COOL$ and $\OciorCOOL$ \cite{Chen:2020arxiv, ChenDISC:21, ChenOciorCOOL:24}. Here,  we use notations consistent with that previously used for $\COOL$ and $\OciorCOOL$. 
Here we use $\Me_i$ to denote the initial value of Node~$i$.  
If Node~$i$ never sets an initial value, then $\Me_i$ is considered as $\Me_i = \defaultvalue$, where $\defaultvalue$ denotes  a default value or an empty value. 
The values of $\Ry_i^{[1]}$ and $\Ry_i^{[2]}$ are initially set to $\Ry_i^{[1]} = \defaultvalue$ and $\Ry_i^{[2]} = \defaultvalue$, respectively.  
The set $\Fc$ is defined as the set of indices of all dishonest nodes.  
We let $[n] := \{1, 2, \dots, n\}$.  
We define some groups of honest nodes as
  \begin{align}
 \Ac_{\Lin} \defeqnew &  \{  i \in[n]\setminus \Fc \mid    \Me_i =  \Meg_{\Lin}  \}, \quad   \Lin \in [\gn]     \label{eq:ACOOLAell00}   \\    
\Ac_{\Lin}^{[p]} \defeqnew&   \{  i \in[n]\setminus \Fc \mid     \Me_i =  \Meg_{\Lin},   \  \Ry_i^{[p]} =1 \}, \quad   \Lin \in [ \gn^{[p]}], \quad  p\in [2]      \label{eq:ACOOLAell}   \\
\AsetNOzero_{\Lin}^{[p]} \defeqnew&   \{  i \in[n]\setminus \Fc \mid     \Me_i =  \Meg_{\Lin},   \  \Ry_i^{[p]} \neq 0 \}, \quad   \Lin \in [ \gnNozero^{[p]}], \quad  p\in [2]      \label{eq:ACOOLANoZero}   \\
\Bc^{[p]} \defeqnew  & \{  i \in[n]\setminus \Fc \mid  \Ry_i^{[p]} =0   \}, \quad  p\in [2]    \label{eq:ACOOLBdef01} \\
\BsetNOone^{[p]} \defeqnew  & \{  i \in[n]\setminus \Fc \mid   \Ry_i^{[p]} \neq 1    \}, \quad  p\in [2]    \label{eq:ACOOLBNoOne} 
 \end{align}  
 for  some  non-empty $\ell$-bit distinct values $\Meg_{1}, \Meg_{2}, \cdots, \Meg_{\gn}$ and some non-negative integers $\gn, \gn^{[1]},\gn^{[2]}, \gnNozero^{[1]},\gnNozero^{[2]}$  such that $ \gn^{[2]} \leq \gn^{[1]} \leq \gn$,  $ \gnNozero^{[2]} \leq \gnNozero^{[1]} \leq \gn$,   $ \gn^{[1]} \leq   \gnNozero^{[1]}$, and  $ \gn^{[2]} \leq   \gnNozero^{[2]}$.  
 In the above definitions, we  consider the final value of $\Ry_i^{[p]}$ if it has been updated. 
Group $\Ac_{\Lin}$ (and Groups $\Ac_{\Lin}^{[p]}$, $\AsetNOzero_{\Lin}^{[p]}$) can be divided into some possibly overlapping sub-groups  defined as 
    \begin{align}
 \Ac_{\Lin,j} \defeqnew   & \{   i\in  \Ac_{\Lin}  \mid  \hv_i^\T  \Meg_{\Lin}  = \hv_i^\T  \Meg_j \} , \quad  j\neq \Lin ,   \  j, \Lin \in [\gn]   \label{eq:Alj}  \\     
 \Ac_{\Lin,\Lin} \defeqnew  &  \Ac_{\Lin}   \setminus  \bigl(\cup_{j=1, j\neq \Lin}^{\gn}\Ac_{\Lin,j}\bigr)   , \quad \quad \quad\quad \Lin \in [\gn]   \label{eq:All}       \\
 \Ac_{\Lin,j}^{[p]} \defeqnew   & \{  i\in  \Ac_{\Lin}^{[p]} \mid  \hv_i^\T  \Meg_{\Lin}  = \hv_i^\T  \Meg_j \} , \quad  j\neq \Lin ,   \  j, \Lin  \in [ \gn^{[p]}], \quad  p\in [2]  \label{eq:Alj11}  \\     
 \Ac_{\Lin,\Lin}^{[p]} \defeqnew  &  \Ac_{\Lin}^{[p]}   \setminus  \bigl(\cup_{j=1, j\neq \Lin}^{\gn^{[p]}}\Ac_{\Lin,j}^{[p]}\bigr)  , \quad\quad\quad\quad  \Lin \in [ \gn^{[p]}], \quad  p\in[2]   \label{eq:All11}    \\
  \AsetNOzero_{\Lin,j}^{[p]} \defeqnew   & \{  i\in  \AsetNOzero_{\Lin}^{[p]} \mid  \hv_i^\T  \Meg_{\Lin}  = \hv_i^\T  \Meg_j \} , \quad  j\neq \Lin ,   \  j, \Lin  \in [ \gnNozero^{[p]}], \quad  p\in [2]  \label{eq:Alj22}  \\     
 \AsetNOzero_{\Lin,\Lin}^{[p]} \defeqnew  &  \AsetNOzero_{\Lin}^{[p]}   \setminus \bigl(\cup_{j=1, j\neq \Lin}^{\gnNozero^{[p]}}\AsetNOzero_{\Lin,j}^{[p]}\bigr)  , \quad\quad\quad\quad  \Lin \in [ \gnNozero^{[p]}], \quad  p\in [2]   \label{eq:ACOOLAll22}         
 \end{align}  
 where $\hv_i$ is the encoding vector of the error-correction code, such that the $i$-th encoded symbol is computed as $\EncodedSymbol_i = \hv_i^\top \MVBAInputMsg$, given the input vector $\MVBAInputMsg$, for $i\in[n]$. 
If a non-linear encoding is used, then the $i$-th encoded symbol is computed as $\EncodedSymbol_i = \ECCEncodingFunction_i(\MVBAInputMsg)$, where $\ECCEncodingFunction_i(\bullet)$ is a polynomial function evaluated at the $i$-th point. In this case, we can simply replace the term $\hv_i^\top \Meg_{\Lin}$ with $\ECCEncodingFunction_i(\MVBAInputMsg)$. For consistency, however, we retain the original definition with linear encoding.

We use $\Lk_i(j) \in \{0,1\}$ to denote the link indicator between Node~$i$ and Node~$j$, defined as
\begin{numcases}
{ \Lk_i(j) =} 
1 & if $(y_i^{(j)}, y_j^{(j)}) = (y_i^{(i)}, y_j^{(i)})$,\label{eq:lkindicator} \\
0 & otherwise, \nonumber
\end{numcases} 
(see  Line~\ref{line:ACOOLph1MatchCond} of Algorithm~\ref{algm:COOLBUA}). In our setting, it holds true that $\Lk_i(j) =  \Lk_j(i)$ for  any $i, j \in [n]\setminus \Fc$.   If $(y_i^{(i)}, y_j^{(i)})$  are never sent out by Node~$i$ or  $(y_i^{(j)}, y_j^{(j)})$ are never sent out by Node~$j$, then we consider $\Lk_i(j) =  \Lk_j(i) = 0$.

Here we let  
\begin{align}
\Lkset_1^{(i)}:=\{j\in [n] \mid   (y_i^{(j)}, y_j^{(j)}) = (y_i^{(i)}, y_j^{(i)})  \}, \quad \text{and} \quad  \Lkset_0^{(i)}:=\{j\in [n] \mid   (y_i^{(j)}, y_j^{(j)}) \neq  (y_i^{(i)}, y_j^{(i)})  \}     \label{eq:LksetUseri}
\end{align}
denote the  link indicator sets $\Lkset_1$ and $\Lkset_0$ updated by Node~$i$, as described in Line~\ref{line:ACOOLph1MatchCond}  of Algorithm~\ref{algm:COOLBUA}. Here  $y_i^{(j)}=\hv_i^\T \Me_{j}$ and $y_j^{(j)}=\hv_j^\T \Me_{j}$ are the $i$-th and the $j$-th coded symbols encoded from the input message $\Me_{j}$ of Node~$j$.  We define 
\begin{align}
\LksetHonest_1^{(i)} := \Lkset_1^{(i)} \setminus \Fc, \quad \text{and}  \quad \LksetHonest_0^{(i)} := \Lkset_0^{(i)} \setminus \Fc.   \label{eq:LksetHonestUseri}
\end{align}

Theorems~\ref{thm:OciorACOOLConsistency}-\ref{thm:OciorACOOLPerformance} present the main results of $\OciorACOOL$. 
Specifically, Theorems~\ref{thm:OciorACOOLConsistency}-\ref{thm:OciorACOOLTermination} show that $\OciorACOOL$ satisfies the Consistency, Validity, and Termination properties in all  executions (\emph{error-free}), under the assumption of an error-free $\ABBA$ invoked in this protocol.
From Theorem~\ref{thm:OciorACOOLPerformance}, it reveals  that $\OciorACOOL$ is optimal   in terms of  communication complexity, round complexity and resilience.

 \begin{theorem}  [Consistency]  \label{thm:OciorACOOLConsistency}
In $\OciorACOOL$,  given $n\geq 3t+1$, if one  honest node outputs a value $\Me^{\star}$, then every honest node  eventually outputs a value $\Me^{\star}$, for some $\Me^{\star}$. 
\end{theorem}
\begin{proof}
In $\OciorACOOL$, if an honest node outputs a value $\Me^{\star}$, it must have set the variable $\VrOutput$ in Line~\ref{line:OciorACOOLVrOutput} of Algorithm~\ref{algm:OciorACOOL}. 
From Lemma~\ref{lm:OciorACOOLvoutput}, if an honest node sets $\VrOutput = \Vr^\star$ in Line~\ref{line:OciorACOOLVrOutput} of Algorithm~\ref{algm:OciorACOOL}, for $\Vr^\star \in \{0,1\}$, then all honest nodes eventually set $\VrOutput = \Vr^\star$. 
Therefore, if one honest node sets $\VrOutput = 0$, then all honest nodes eventually set $\VrOutput = 0$ and output $\defaultvalue$ in Line~\ref{line:OciorACOOLph3} of Algorithm~\ref{algm:OciorACOOL}.
In the following, we focus on the case where all honest nodes eventually set $\VrOutput = 1$.

When an honest node sets $\VrOutput = 1$, it implies that $\ABBA$ must have output the value $1$ (see Lines~\ref{line:OciorACOOLRBAbegin}-\ref{line:OciorACOOLVrOutput} of Algorithm~\ref{algm:OciorACOOL}). 
Due to the \emph{Validity} and \emph{Consistency} properties of $\ABBA$, if $\ABBA$ outputs the value $1$, then at least one honest node~$i$ has input the value $1$ to $\ABBA$ and has delivered $\VrSecond_i = 1$ from $\COOLBUA[2]$ (see Lines~\ref{line:OciorACOOLABBAinputCond} and \ref{line:OciorACOOLABBAinput} of Algorithm~\ref{algm:OciorACOOL}). 
When an honest node~$i$ has delivered $\VrSecond_i = 1$ from $\COOLBUA[2]$, at least $n - 2\networkfaultsizet$ honest nodes must have set $\RySecond_j^{[2]} = 1$ in $\COOLBUA[2]$ (see Lines~\ref{line:OciorACOOLReadyOneCondition} and \ref{line:OciorACOOLReadyOne} of Algorithm~\ref{algm:COOLBUA}).

From Lemma~\ref{lm:OciorACOOBUAUniqueSet}, the honest nodes that set $\RySecond_i^{[2]} = 1$ in $\COOLBUA[2]$ must have the same input message $\MeSecond^{(i)} = \Me^{\star}$, for some $\Me^{\star}$. 
If an honest node~$i$ has set $\VrOutput = 1$ and has delivered $[\MeSecond^{(i)}, \RySecond_{i}^{[2]} = 1, \VrSecond_i]$ from $\COOLBUA[2]$, then this node outputs the value $\MeSecond^{(i)} = \Me^{\star}$ (see Line~\ref{line:OciorACOOLph3SIPhtwoOutput} of Algorithm~\ref{algm:OciorACOOL}).

If an honest node~$i$ has set $\VrOutput = 1$ but has not yet delivered $[\MeSecond^{(i)}, \RySecond_{i}^{[2]} = 1, \VrSecond_i]$ from $\COOLBUA[2]$, it can be shown that this node will eventually output the same value $\Me^{\star}$ in Line~\ref{line:OciorACOOLph3Output2} of Algorithm~\ref{algm:OciorACOOL}. 
Recall that when an honest node has delivered $\VrSecond = 1$ from $\COOLBUA[2]$, at least $n - 2\networkfaultsizet \geq t + 1$ honest nodes must have set $\RySecond_i^{[2]} = 1$ in $\COOLBUA[2]$. 
In this case, from Lemma~\ref{lm:OciorACOOBUAUniqueSet}, at least $t + 1$ honest nodes~$i$ must have set $\RySecond_i^{[2]} = 1$ and  have sent $\ltuple \SYMBOL, 2, (\ysymbolSecond_j^{(i)}(\Me^{\star}), \ysymbolSecond_i^{(i)}(\Me^{\star})) \rtuple$ to Node~$j$, for all $j \in [\networksizen]$, where $\ysymbolSecond_j^{(i)}(\Me^{\star})$ and $\ysymbolSecond_i^{(i)}(\Me^{\star})$ denote the coded symbols generated from the message $\Me^{\star}$. 

From the above results, if an honest node~$i$ has set $\VrOutput = 1$ but has not yet delivered $[\MeSecond^{(i)}, \RySecond_{i}^{[2]} = 1, \VrSecond_i]$ from $\COOLBUA[2]$, then it will eventually deliver from $\COOLBUA[2]$ at least $t + 1$ matching $\ltuple \SYMBOL, 2, (\ysymbolSecond_i^{(j)}(\Me^{\star}), *) \rtuple$ messages for all $j \in \SsSecond_1^{[2]}$, where $\ysymbolSecond_i^{(j)}(\Me^{\star}) = y^{\star}_i$ for some $y^{\star}_i$, which is the $i$-th symbol encoded from the message $\Me^{\star}$. 
In this case, Node~$i$ sets $\ysymbolSecond_i^{(i)} = y^{\star}_i$ in Line~\ref{line:OciorACOOLph3MajorityRule} and sends $\ltuple \CORRECTSYMBOL, \IDCOOL, \ysymbolSecond_i^{(i)} \rtuple$ to all nodes in Line~\ref{line:OciorACOOLph3SendCorrectSymbols}. 
Thus, each symbol $y_j^{(j)}$ included in $\OECCorrectSymbolSet \setminus \Fc$ must be encoded from the same message $\Me^{\star}$. 
Therefore, if an honest node~$i$ has set $\VrOutput = 1$ but has not yet delivered $[\MeSecond^{(i)}, \RySecond_{i}^{[2]} = 1, \VrSecond_i]$ from $\COOLBUA[2]$, it will eventually decode the message $\Me^{\star}$ using online error-correction decoding and output $\Me^{\star}$ in Line~\ref{line:OciorACOOLph3Output2} of Algorithm~\ref{algm:OciorACOOL}. 
 \end{proof}

 \begin{theorem}  [Validity]  \label{thm:OciorACOOLvalidity}   
Given $n\geq 3t+1$,    if all honest nodes input the same  value $\wv$, then every honest node eventually outputs $\wv$ in $\OciorACOOL$.     
\end{theorem}
\begin{proof} 
From Theorem~\ref{thm:OciorACOOLTermination}, if all honest nodes receive their inputs, then every honest node eventually outputs a value and terminates in $\OciorACOOL$. 
Furthermore, from Theorem~\ref{thm:OciorACOOLConsistency}, in $\OciorACOOL$, if one honest node outputs a value $\Me^{\star}$, then every honest node eventually outputs the same value $\Me^{\star}$, for some $\Me^{\star}$. 
Thus, based on Theorem~\ref{thm:OciorACOOLTermination} and Theorem~\ref{thm:OciorACOOLConsistency}, what remains to prove for this theorem is that if all honest nodes input the same value $\wv$, and if an honest node~$i$ outputs a value $\Me^{\star}$ in $\OciorACOOL$, then $\Me^{\star} = \wv$.

First, we prove that if all honest nodes input the same message $\wv$ in $\OciorACOOL$, then in $\COOLBUA[1]$, no honest node will set $\Ry_i^{[1]} = 0$ (see Lines~\ref{line:COOLBUASI1zeroCond} and \ref{line:COOLBUASI1zero} of Algorithm~\ref{algm:COOLBUA}), no honest node will set $\Ry_i^{[2]} = 0$ (see Lines~\ref{line:ACOOLph2CSI0Cond} and \ref{line:ACOOLph2CSI0} of Algorithm~\ref{algm:COOLBUA}), and no honest node will set $\Vr_i = 0$ (see Lines~\ref{line:OciorACOOLReadyZeroCondition} and \ref{line:OciorACOOLReadyZero} of Algorithm~\ref{algm:COOLBUA}). 
Specifically, in this case, if an honest node~$i$ sends out the message in Line~\ref{line:ACOOLECCEncindicator} of Algorithm~\ref{algm:COOLBUA}, then the message must be $\ltuple\SYMBOL, \IDCOOL, (y_j^{(i)}(\wv), y_i^{(i)}(\wv))\rtuple$, where $y_j^{(i)}(\wv)$ and $y_i^{(i)}(\wv)$ are symbols encoded from the message $\wv$ (see Lines~\ref{line:ACOOLECCEnc} and \ref{line:ACOOLECCEncindicator} of Algorithm~\ref{algm:COOLBUA}). 
Thus, the equality $(y_i^{(j)}(\wv), y_j^{(j)}(\wv)) = (y_i^{(i)}(\wv), y_j^{(i)}(\wv))$ must hold for any $i, j \in [n] \setminus \Fc$ (see Line~\ref{line:ACOOLph1MatchCond} of Algorithm~\ref{algm:COOLBUA}). 
This implies that the condition $|\Lkset_0| \geq \networkfaultsizet + 1$ in Line~\ref{line:COOLBUASI1zeroCond} of Algorithm~\ref{algm:COOLBUA} will never be satisfied at any honest node, and therefore no honest node will set $\Ry_i^{[1]} = 0$ in Line~\ref{line:COOLBUASI1zero} of Algorithm~\ref{algm:COOLBUA}. 
Since no honest node will send out $\ltuple\SIone, \IDCOOL, \Ry_i^{[1]} = 0\rtuple$, the condition $|\Ss_0^{[1]}| \geq \networkfaultsizet + 1$ in Line~\ref{line:COOLBUASI1zeroCond} of Algorithm~\ref{algm:COOLBUA} will never be satisfied at any honest node, and thus no honest node will set $\Ry_i^{[2]} = 0$ in Line~\ref{line:ACOOLph2CSI0} of Algorithm~\ref{algm:COOLBUA}. 
Similarly, since no honest node will send out $\ltuple\SItwo, \IDCOOL, \Ry_i^{[2]} = 0\rtuple$, the condition $|\Ss_0^{[2]}| \geq \networkfaultsizet + 1$ in Line~\ref{line:OciorACOOLReadyZeroCondition} of Algorithm~\ref{algm:COOLBUA} will never be satisfied at any honest node, and hence no honest node will set $\Vr_i = 0$ in Line~\ref{line:OciorACOOLReadyZero} of Algorithm~\ref{algm:COOLBUA}.

Next, we prove that if all honest nodes input the same message $\wv$ in $\OciorACOOL$, and if any honest node~$i$ passes $\MeSecond_i$ into $\COOLBUA[2]$ as an input message (see Lines~\ref{line:BUA2Input1} and \ref{line:BUA2Input2} of Algorithm~\ref{algm:OciorACOOL}), then $\MeSecond_i = \wv$. 
Specifically, the above results show that, in this case, no honest node will set $\RyNew_i^{[2]} = 0$. 
If any honest node~$i$ passes $\MeSecond_i$ into $\COOLBUA[2]$ as an input message, and if node~$i$ has set $\RyNew_i^{[2]} = 1$ from $\COOLBUA[1]$, then  node~$i$ sets $\MeSecond_i = \Me_i = \wv$ and passes $\wv$ into $\COOLBUA[2]$ as its input message (see Lines~\ref{line:BUA2Input1Cond}-\ref{line:BUA2Input1} of Algorithm~\ref{algm:OciorACOOL}). 
Furthermore, from Lemma~\ref{lm:OciorACOOLSameinputOECdec}, if all honest nodes input the same value $\wv$, and if an honest node~$i$ sets $\MeSecond_i = \MVBAOutputMsg$ in Line~\ref{line:BUA2OECend} of Algorithm~\ref{algm:OciorACOOL}, where $\MVBAOutputMsg$ is the message decoded from online error correction in Line~\ref{line:BUA1OECdec}, then $\MVBAOutputMsg = \wv$. This results implies that,   if all honest nodes input the same value $\wv$, and if  any honest node~$i$ passes $\MeSecond_i$ into $\COOLBUA[2]$ as an input message, then $\MeSecond_i =  \wv$ (see Lines~\ref{line:BUA2Input2Cond}-\ref{line:BUA2Input2} of Algorithm~\ref{algm:OciorACOOL}).

From the above results, if all honest nodes input the same message $\wv$ in $\OciorACOOL$, then the input messages (if any) of all honest nodes in $\COOLBUA[2]$ must also be $\wv$. 
Following from the previous arguments, in this case, within $\COOLBUA[2]$, no honest node will set $\RySecond_i^{[1]} = 0$ (see Lines~\ref{line:COOLBUASI1zeroCond} and \ref{line:COOLBUASI1zero} of Algorithm~\ref{algm:COOLBUA}), no honest node will set $\RySecond_i^{[2]} = 0$ (see Lines~\ref{line:ACOOLph2CSI0Cond} and \ref{line:ACOOLph2CSI0} of Algorithm~\ref{algm:COOLBUA}), and no honest node will set $\VrSecond_i = 0$ (see Lines~\ref{line:OciorACOOLReadyZeroCondition} and \ref{line:OciorACOOLReadyZero} of Algorithm~\ref{algm:COOLBUA}).

From Theorem~\ref{thm:OciorACOOLTermination}, if all honest nodes input the same message $\wv$ in $\OciorACOOL$, then every honest node eventually outputs a value and terminates in $\OciorACOOL$. 
In this case, by combining the above results, we conclude that   no honest node sets $\VrSecond_i = 0$,  and that at least one honest node~$i$ delivers $[\MeSecond^{(i)} = \wv, \RySecond_{i}^{[2]}, \VrSecond_i = 1]$ from $\COOLBUA[2]$ and inputs $1$ into $\ABBA$ as its input message in Line~\ref{line:OciorACOOLABBAinput} of Algorithm~\ref{algm:OciorACOOL}. 
Since no honest node sets $\VrSecond_i = 0$, no honest node will input $0$ into $\ABBA$ as its input message. 
This implies that $\ABBA$ eventually outputs $1$, and all honest nodes eventually set $\VrOutput = 1$ (see Lines~\ref{line:OciorACOOLRBAbegin}-\ref{line:OciorACOOLph3} of Algorithm~\ref{algm:OciorACOOL}). 

Consequently, if all honest nodes input the same message $\wv$ in $\OciorACOOL$, and if an honest node~$i$ has set $\VrOutput = 1$ and has delivered $[\MeSecond^{(i)}, \RySecond_{i}^{[2]} = 1, \VrSecond_i]$ from $\COOLBUA[2]$, then node~$i$ eventually outputs the value $\wv$ in Line~\ref{line:OciorACOOLph3SIPhtwoOutput} of Algorithm~\ref{algm:OciorACOOL}. 
If an honest node~$i$ has set $\VrOutput = 1$ but has not yet delivered $[\MeSecond^{(i)}, \RySecond_{i}^{[2]} = 1, \VrSecond_i]$ from $\COOLBUA[2]$,   this node   eventually outputs the same value $\wv$ in Line~\ref{line:OciorACOOLph3Output2} of Algorithm~\ref{algm:OciorACOOL} (see the last two paragraphs of the proof of Theorem~\ref{thm:OciorACOOLConsistency}). 
This completes the proof.
\end{proof}

\begin{lemma}     \label{lm:OciorACOOLSameinputOECdec}
In $\OciorACOOL$, if all honest nodes input the same value $\wv$, and if an honest node~$i$ sets $\MeSecond_i = \MVBAOutputMsg$ in Line~\ref{line:BUA2OECend} of Algorithm~\ref{algm:OciorACOOL}, where $\MVBAOutputMsg$ is the message decoded from online error correction in Line~\ref{line:BUA1OECdec}, then $\MVBAOutputMsg = \wv$. 
\end{lemma}
\begin{proof}
We consider the case where all honest nodes input the same value $\wv$. 
In this case, if there exists a symbol $y^{\star}$ such that the condition $|\ySymbolSet[y^{\star}]| \geq n - 2t \geq t + 1$ in Line~\ref{line:BUA2InputYiMajorityCond} of Algorithm~\ref{algm:OciorACOOL} is satisfied at an honest node~$i$, then it follows that $y^{\star} = y_i(\wv)$, i.e., $y^{\star}$ is equal to the $i$-th coded symbol encoded from the message $\wv$. 
Thus, in this case, if any honest node~$i$ sends $\ltuple \NEWSYMBOL, \IDCOOL, \ystrongmajority_i \rtuple$ to all nodes in Line~\ref{line:BUA2InputYiMajoritySend}, then $\ystrongmajority_i = y_i(\wv)$. 
Furthermore, in this case, if any honest node~$i$ sends $\ltuple \SYMBOL, 1, (*, \ysymbolNew_i^{(i)}) \rtuple$ in $\COOLBUA[1]$, then $\ysymbolNew_i^{(i)} = y_i(\wv)$. 
Therefore, if an honest node sets $\OECCorrectSymbolSetNew[j] \gets \ystrongmajority_j$ in Lines~\ref{line:BUA2InputYiMajorityyj} and~\ref{line:BUA2InputYiMajoritySi1} of Algorithm~\ref{algm:OciorACOOL}, for any $j \in [n] \setminus \Fc$, then $\ystrongmajority_j = y_j(\wv)$. 
In other words, each symbol $y_j^{(j)}$ included in $\OECCorrectSymbolSetNew$ at any honest node~$i$, for any $j \in [n] \setminus \Fc$, must be encoded from the same message $\wv$. 
Hence, in this case, if any honest node~$i$ decodes a message $\MVBAOutputMsg$ from online error correction in Line~\ref{line:BUA1OECdec}, then $\MVBAOutputMsg = \wv$. 
\end{proof}

 \begin{theorem}  [Termination]  \label{thm:OciorACOOLTermination}   
Given $n\geq 3t+1$,    if all  honest nodes receive their inputs, then every honest node  eventually outputs a value and terminates  in $\OciorACOOL$.     
\end{theorem}
\begin{proof}
From Lemma~\ref{lm:OciorACOOLTerminationRBA}, if all honest nodes receive their inputs, then at least one honest node eventually sets $\VrOutput = \Vr^\star$ in Line~\ref{line:OciorACOOLVrOutput} of Algorithm~\ref{algm:OciorACOOL}, for some $\Vr^\star \in \{0,1\}$. 
From Lemma~\ref{lm:OciorACOOLvoutput}, if an honest node sets $\VrOutput = \Vr^\star$ in Line~\ref{line:OciorACOOLVrOutput} of Algorithm~\ref{algm:OciorACOOL}, then all honest nodes eventually set $\VrOutput = \Vr^\star$. 
If one honest node sets $\VrOutput = 0$, then all honest nodes eventually set $\VrOutput = 0$, output $\defaultvalue$, and terminate in Line~\ref{line:OciorACOOLph3} of Algorithm~\ref{algm:OciorACOOL}. 
In the following, we focus on the case where all honest nodes eventually set $\VrOutput = 1$.

As in the proof of Theorem~\ref{thm:OciorACOOLConsistency}, if an honest node~$i$ has set $\VrOutput = 1$ and has delivered $[\MeSecond^{(i)}, \RySecond_{i}^{[2]} = 1, \VrSecond_i]$ from $\COOLBUA[2]$, then this node outputs the value $\MeSecond^{(i)}$ and terminates in Line~\ref{line:OciorACOOLph3SIPhtwoOutput} of Algorithm~\ref{algm:OciorACOOL}.  
It has also been shown in the proof of Theorem~\ref{thm:OciorACOOLConsistency} that if an honest node~$i$ has set $\VrOutput = 1$ but has not yet delivered $[\MeSecond^{(i)}, \RySecond_{i}^{[2]} = 1, \VrSecond_i]$ from $\COOLBUA[2]$, then Node~$i$ eventually outputs a message  and terminates in Line~\ref{line:OciorACOOLph3Output2} of Algorithm~\ref{algm:OciorACOOL} (see the last two paragraphs of the proof of Theorem~\ref{thm:OciorACOOLConsistency}). 
\end{proof}

\begin{lemma}  [Totality and Consistency Properties of $\RBA$]  \label{lm:OciorACOOLvoutput}
In  $\OciorACOOL$,  if an  honest node sets     $\VrOutput= \Vr^\star$   in  Line~\ref{line:OciorACOOLVrOutput}  of Algorithm~\ref{algm:OciorACOOL}, for $\Vr^\star\in \{1,0\}$, then  all honest nodes  eventually set  $\VrOutput= \Vr^\star$. 
\end{lemma}
\begin{proof}
This result follows from the \emph{Totality} and \emph{Consistency} properties of $\RBA$, as described in Lines~\ref{line:OciorACOOLRBAbegin}-\ref{line:OciorACOOLph3} of Algorithm~\ref{algm:OciorACOOL}.        
Specifically, when an honest node sets $\VrOutput = \Vr^\star$ in Line~\ref{line:OciorACOOLVrOutput} of Algorithm~\ref{algm:OciorACOOL}, for $\Vr^\star \in \{1,0\}$, it implies that this node has received at least $2t + 1$ $\ltuple \READY, \IDCOOL, \Vr^\star \rtuple$ messages (see Line~\ref{line:OciorACOOLVrOutputCond}). 
In this case, at least $t + 1$ honest nodes must have sent out $\ltuple \READY, \IDCOOL, \Vr^\star \rtuple$ messages. 

In our setting, if an honest node has sent out a $\ltuple \READY, \IDCOOL, \Vr^\star \rtuple$ message, then $\ABBA$ must have output a value $\ABBAOutput$ at at least one honest node (see Line~\ref{line:OciorACOOLRBAbegin}).   
Due to the \emph{Consistency} property of $\ABBA$, if two messages $\ltuple \READY, \IDCOOL, \Vr^\star \rtuple$ and $\ltuple \READY, \IDCOOL, \Vr' \rtuple$ are sent out from two honest nodes, respectively, then $\Vr^\star = \Vr'$.  

Therefore,  if an honest node sets $\VrOutput = \Vr^\star$, then each honest node eventually sends out a message $\ltuple \READY, \IDCOOL, \Vr^\star \rtuple$ (see Lines~\ref{line:OciorACOOLRelialbeBegin} and~\ref{line:OciorACOOLRBAsendReady}). 
Consequently, each honest node eventually receives at least $2t + 1$ $\ltuple \READY, \IDCOOL, \Vr^\star \rtuple$ messages and subsequently sets $\VrOutput = \Vr^\star$ in Line~\ref{line:OciorACOOLVrOutput}.  
Note that every honest node should have sent out a message $\ltuple \READY, \IDCOOL, \Vr^\star \rtuple$ and set $\VrOutput = \Vr^\star$ before termination.  
\end{proof}

\begin{lemma}  [Termination]  \label{lm:OciorACOOLTerminationRBA}
In  $\OciorACOOL$, if all  honest nodes receive their inputs, then at least  an  honest node eventually sets     $\VrOutput= \Vr^\star$   in  Line~\ref{line:OciorACOOLVrOutput}  of Algorithm~\ref{algm:OciorACOOL}, for some $\Vr^\star\in \{1,0\}$. 
\end{lemma}
\begin{proof}
We prove this result by contradiction. 
Assume that no honest node sets the value of $\VrOutput$ in Line~\ref{line:OciorACOOLVrOutput} of Algorithm~\ref{algm:OciorACOOL}. 
Under this assumption, no honest node would terminate in $\OciorACOOL$. 
Note that every honest node must set $\VrOutput = \Vr^\star$ before termination (see Lines~\ref{line:OciorACOOLVrOutput} and \ref{line:OciorACOOLph3} of Algorithm~\ref{algm:OciorACOOL}).

From Lemma~\ref{lm:OciorACOOBUAUniqueSet},  if all honest nodes eventually input their initial messages and keep running the    $\COOLBUA[1]$ protocol,  then  in $\COOLBUA[1]$  it holds true that $\gnNozero^{[2]} \leq    1$, i.e., every honest node~$i$ eventually sets 
\begin{align}
  \Ry_i^{[2]} = 0,  \quad \forall  i \in \Ac_{\Lin},  \forall \Lin \in [2, \gn]   \label{eq:Bzerory}  
 \end{align}
  in $\COOLBUA[1]$. 
Here, $\AsetNOzero_{1}^{[2]} \defeqnew   \{  i \in[n]\setminus \Fc \mid     \Me_i =  \Meg_{1},   \  \Ry_i^{[2]} \neq 0 \}$    (see \eqref{eq:ACOOLANoZero}).   
 From the  result in \eqref{eq:Bzerory}, and based on our notation $\Bc^{[2]} \defeqnew  \{  i \in[n]\setminus \Fc \mid  \Ry_i^{[2]} =0   \}$  (see \eqref{eq:ACOOLBdef01}), it is true that   
\begin{align}
|\AsetNOzero_{1}^{[2]}|  +  |\Bc^{[2]}| =n-|\Fc|    \label{eq:Sizeidentity}  
 \end{align}
and that 
\begin{align}
 \bigcup_{\Lin \in [2, \gn]  } \Ac_{\Lin}  \subseteq \Bc^{[2]}    \non
 \end{align}
 for $\COOLBUA[1]$. 
  In the following, we   consider each of the two cases for $\COOLBUA[1]$: 
 \begin{align}
\text{Case~I:} \quad  |\AsetNOzero_{1}^{[2]}| \geq     n- |\Fc| -t ,     \label{eq:caseI} 
 \end{align}
\begin{align}
\text{Case~II:} \quad  |\AsetNOzero_{1}^{[2]}| <    n- |\Fc| -t  .   \label{eq:caseII}  
 \end{align}
 
 \vspace{8pt}
 \noindent \emph{$\bullet$ \underline{Analysis for Case~I:}}  \\
 Let us first consider Case~I in \eqref{eq:caseI} with  $|\AsetNOzero_{1}^{[2]}| \geq     n- |\Fc| -t$ for $\COOLBUA[1]$. 
 In this case,  from  Lemma~\ref{lm:OciorACOOLCaseISameInputBUA2},   if all honest nodes eventually input their initial messages and keep running the $\OciorACOOL$, $\COOLBUA[1]$ and $\COOLBUA[2]$   protocols, then   every honest node~$i$ eventually inputs the same  message $\MeSecond_i=\Meg_{1}$ to $\COOLBUA[2]$, where $\Meg_{1}$ is the initial message of all honest nodes within $\AsetNOzero_{1}^{[2]}$.  
 
If all honest nodes  keep running the  $\OciorACOOL$, $\COOLBUA[1]$ and $\COOLBUA[2]$   protocols,  and  given the conclusion that  every honest node~$i$ eventually inputs the same  message $\MeSecond_i=\Meg_{1}$ to $\COOLBUA[2]$, then  every honest node~$i$ eventually sets $\RySecond_i^{[1]} =1$ (see Lines~\ref{line:COOLBUASISI1oneCond} and \ref{line:COOLBUASISI1oneSend} of Algorithm~\ref{algm:COOLBUA}),
every honest node~$i$ eventually sets $\RySecond_i^{[2]} =1$ (see Lines~\ref{line:ACOOLph2ACond} and \ref{line:ACOOLph2ASI1} of Algorithm~\ref{algm:COOLBUA}),
 and every honest node~$i$ eventually sets  $\VrSecond_i = 1$ (see Lines~\ref{line:OciorACOOLReadyOneCondition} and \ref{line:OciorACOOLReadyOne} of Algorithm~\ref{algm:COOLBUA}).
  In this case,  every honest node~$i$ eventually   delivers $[\MeSecond^{(i)}, \RySecond_{i}^{[2]}, \VrSecond_i = 1]$ from $\COOLBUA[2]$. 
  
Consequently, in this case, every honest node~$i$ eventually passes $1$ into $\ABBA$ as its input message, if node~$i$ has not already provided an input to $\ABBA$ (see Lines~\ref{line:OciorACOOLABBAinputCond} and~\ref{line:OciorACOOLABBAinput} of Algorithm~\ref{algm:OciorACOOL}). 
Hence, $\ABBA$ eventually outputs a value $\ABBAOutput \in \{1, 0\}$ at each node, due to the \emph{Termination} property of $\ABBA$. 
Then, every honest node eventually sends $\ltuple \READY, \IDCOOL, \ABBAOutput \rtuple$ to all nodes in Line~\ref{line:OciorACOOLRBAReadyABBA} of Algorithm~\ref{algm:OciorACOOL}, and every honest node eventually sets $\VrOutput = \Vr^\star$ in Line~\ref{line:OciorACOOLVrOutput} of Algorithm~\ref{algm:OciorACOOL}. 
This conclusion contradicts the original assumption that no honest node sets the value of $\VrOutput$ in Line~\ref{line:OciorACOOLVrOutput}. 
Thus, the original assumption is false, and the statement of this lemma holds for Case~I.

 \vspace{8pt}
 
  \noindent \emph{$\bullet$ \underline{Analysis for Case~II:}}  \\
  Let us now consider Case~II in \eqref{eq:caseII}. In this case, we have
\begin{align}
  |\Bc^{[2]}| &= n - |\Fc| - |\AsetNOzero_{1}^{[2]}|  \label{eq:Sizeidentity1} \\
              &> n - |\Fc| - (n - |\Fc| - t)  \label{eq:Sizeidentity2} \\
              &= t   \label{eq:Sizeidentity3}
\end{align}
where \eqref{eq:Sizeidentity1} follows from \eqref{eq:Sizeidentity}, while \eqref{eq:Sizeidentity2} uses the assumption in \eqref{eq:caseII} for this Case~II. 
The result in \eqref{eq:Sizeidentity3} implies that $|\Bc^{[2]}| \geq t + 1$ in $\COOLBUA[1]$. 
In this case, the condition $|\Ss_{0}^{[2]}| \geq t + 1$ in Line~\ref{line:OciorACOOLReadyZeroCondition} of Algorithm~\ref{algm:COOLBUA} is eventually satisfied at all honest nodes, and each honest node~$i$ eventually sets $\VrNew_i = 0$ and delivers $[\MeNew^{(i)}, \RyNew_{i}^{[2]}, \VrNew_i = 0]$ from $\COOLBUA[1]$. 
Consequently, each honest node~$i$ eventually passes $0$ into $\ABBA$ as its input message, if node~$i$ has not already provided an input to $\ABBA$ (see Lines~\ref{line:OciorACOOLABBAinputBUA1Cond} and~\ref{line:OciorACOOLABBAinput1} of Algorithm~\ref{algm:OciorACOOL}). 
In this case, every honest node eventually provides an input to $\ABBA$, and hence $\ABBA$ eventually outputs a value $\ABBAOutput \in \{1, 0\}$ at each node, due to the \emph{Termination} property of $\ABBA$. 
Then, every honest node eventually sends $\ltuple \READY, \IDCOOL, \ABBAOutput \rtuple$ to all nodes in Line~\ref{line:OciorACOOLRBAReadyABBA} of Algorithm~\ref{algm:OciorACOOL}, and every honest node eventually sets $\VrOutput = \Vr^\star$ in Line~\ref{line:OciorACOOLVrOutput} of Algorithm~\ref{algm:OciorACOOL}. 
This conclusion contradicts the original assumption that no honest node sets the value of $\VrOutput$ in Line~\ref{line:OciorACOOLVrOutput}. 
Thus, the original assumption is false, and the statement of this lemma holds for Case~II.
\end{proof}

\begin{lemma}     \label{lm:OciorACOOLCaseISameInputBUA2}
Consider the case where $|\AsetNOzero_{1}^{[2]}| \geq n - |\Fc| - t$ for $\COOLBUA[1]$, and assume that all honest nodes eventually input their initial messages and keep running the $\OciorACOOL$, $\COOLBUA[1]$ and $\COOLBUA[2]$   protocols. 
Under this assumption, every honest node~$i$ eventually inputs the same  message $\MeSecond_i=\Meg_{1}$ to $\COOLBUA[2]$, where $\Meg_{1}$ is the initial message of all honest nodes within $\AsetNOzero_{1}^{[2]}$.  
\end{lemma}
\begin{proof}
Let us assume $|\AsetNOzero_{1}^{[2]}| \geq n - |\Fc| - t$ for $\COOLBUA[1]$, and assume that all honest nodes eventually input their initial messages and keep running the $\OciorACOOL$, $\COOLBUA[1]$ and $\COOLBUA[2]$   protocols. 
Under these assumptions, we    consider  each of the following   cases for every honest node~$i$: 
 \begin{align}
\text{Case~A:} \quad & \text{Node~$i$ sets $\RyNew_i^{[1]} \neq 1$},     \label{eq:caseAnew}  \\ 
\text{Case~B:} \quad & \text{Node~$i$ sets $\RyNew_i^{[1]} =1$}  \label{eq:caseBnew}  
 \end{align}
 where we consider the final value of $\RyNew_i^{[1]}$ updated in  $\COOLBUA[1]$.  
 Under the assumptions considered here, we argue in the following that, for each case, the index of every honest node~$i$ is eventually included in~$\OECCorrectSymbolSetNew$ at 
each honest node, as in Lines~\ref{line:BUA2InputYiMajorityyj} or~\ref{line:BUA2InputYiMajoritySi1} of Algorithm~\ref{algm:OciorACOOL}, such that 
\begin{align}
\OECCorrectSymbolSetNew[i] = y_i(\Meg_{1}) \quad \forall i \in [n] \setminus \Fc .  \label{eq:OECsymboli}  
\end{align}
Here, $y_i(\Meg_{1})$ denotes the $i$-th coded symbol encoded from the message~$\Meg_{1}$, and~$\Meg_{1}$ is the initial message of all honest nodes within~$\AsetNOzero_{1}^{[2]}$.   
With the conclusion in~\eqref{eq:OECsymboli}, it follows that every honest node eventually decodes the message as~$\Meg_{1}$ from the online error correction procedure in Line~\ref{line:BUA1OECdec} of Algorithm~\ref{algm:OciorACOOL}, and then sets the input of~$\COOLBUA[2]$ as  
\begin{align}
\MeSecond_i = \Meg_{1}   \quad \text{(in Line~\ref{line:BUA2OECend} of Algorithm~\ref{algm:OciorACOOL})}.   \label{eq:BUA2input1}  
\end{align} 
If an honest node~$i$ sets $\RyNew_i^{[2]} = 1$ in~$\COOLBUA[1]$ and $\COOLBUA[2]$ has not yet received an input message, then node~$i$ sets $\MeSecond_i = \Me_{i}$ as in Line~\ref{line:BUA2Input1Setw} of Algorithm~\ref{algm:OciorACOOL}.  
Lemma~\ref{lm:OciorACOOBUAUniqueSet} states that if all honest nodes eventually input their initial messages and keep running the~$\COOLBUA[1]$ protocol, then in~$\COOLBUA[1]$ it holds that~$\gnNozero^{[2]} \leq 1$.   
From Lemma~\ref{lm:OciorACOOBUAUniqueSet}, if an honest node~$i$ sets $\RyNew_i^{[2]} = 1$ in~$\COOLBUA[1]$ and $\COOLBUA[2]$ has not yet received an input message, then it follows that $i \in \AsetNOzero_{1}^{[2]}$ and that node~$i$ sets 
\begin{align}
\MeSecond_i =  \Meg_{1}    \quad \text{(in Line~\ref{line:BUA2Input1Setw} of Algorithm~\ref{algm:OciorACOOL})}.   \label{eq:BUA2input2}  
\end{align}

With the results in~\eqref{eq:BUA2input1} and~\eqref{eq:BUA2input2},  it follows that every honest node~$i$ eventually inputs the same message~$\MeSecond_i = \Meg_{1}$ to~$\COOLBUA[2]$ (see Lines~\ref{line:BUA2Input1} and~\ref{line:BUA2Input2} of Algorithm~\ref{algm:OciorACOOL}). 
What remains to be proven now is the conclusion in~\eqref{eq:OECsymboli}.

 \vspace{8pt}

 \noindent \emph{$\bullet$ \underline{Proof of \eqref{eq:OECsymboli}  for Case~A:}}  \\
  We first consider   Case~A, where an honest~$i$ never sets $\RyNew_i^{[1]} =1$.  
 Lemma~\ref{lm:OciorACOOBUAUniqueSet} shows that if all honest nodes eventually input their initial messages and keep running the~$\COOLBUA[1]$ protocol, then it holds that~$\gnNozero^{[2]} \leq 1$ for~$\COOLBUA[1]$.  
Recall that~$\AsetNOzero_{1}^{[2]} \defeqnew \{\, j \in [n] \setminus \Fc \mid \Me_j = \Meg_{1},\ \Ry_j^{[2]} \neq 0 \,\}$ for some~$\Meg_{1}$.  
From Lemma~\ref{lm:OciorACOOBUAUniqueSet}, it follows that the index of each honest node~$j \in [n] \setminus (\Fc \cup \AsetNOzero_{1}^{[2]})$ is eventually included in the set~$\SsNew_0^{[2]}$ at node~$i$, i.e.,
\begin{align}
[n] \setminus (\Fc \cup \AsetNOzero_{1}^{[2]}) \subseteq \SsNew_0^{[2]}. \label{eq:CaseAproof99}
\end{align}
Furthermore, the index of each honest node~$j' \in \AsetNOzero_{1}^{[2]}$ is eventually included in the set~$\ySymbolSet[y_i(\Meg_{1})]$ at node~$i$ (see Lines~\ref{line:BUA2InputSymbolDeliverCond} and~\ref{line:BUA2InputSymbolDeliver} of Algorithm~\ref{algm:OciorACOOL}), i.e.,
\begin{align}
\AsetNOzero_{1}^{[2]} \subseteq \ySymbolSet[y_i(\Meg_{1})]. \label{eq:CaseAproof88}
\end{align}
From~\eqref{eq:CaseAproof88}, and given the condition~$|\AsetNOzero_{1}^{[2]}| \geq n - |\Fc| - t$ considered here for~$\COOLBUA[1]$, the following bound eventually holds:
\begin{align}
|\ySymbolSet[y_i(\Meg_{1})]| \geq |\AsetNOzero_{1}^{[2]}| \geq n - |\Fc| - t \geq n - 2t. \label{eq:CaseAproof77}
\end{align}
Moreover, from~\eqref{eq:CaseAproof99} and~\eqref{eq:CaseAproof88}, we have that the following bound eventually holds:
\begin{align}
|\ySymbolSet[y_i(\Meg_{1})] \cup \SsNew_0^{[2]}|
\geq \big| ([n] \setminus (\Fc \cup \AsetNOzero_{1}^{[2]})) \cup \AsetNOzero_{1}^{[2]} \big| 
= |[n] \setminus \Fc| = n - |\Fc| \geq n - t. \label{eq:CaseAproof66}
\end{align}
Therefore, from~\eqref{eq:CaseAproof77} and~\eqref{eq:CaseAproof66}, if~$\RyNew_i^{[1]} \neq 1$, then the conditions 
\[
|\ySymbolSet[y_i(\Meg_{1})] \cup \SsNew_0^{[2]}| \geq n - t, 
\quad 
|\ySymbolSet[y_i(\Meg_{1})]| \geq n - 2t,
\quad 
\text{and} \quad 
\RyNew_i^{[1]} \neq 1
\]
in Line~\ref{line:BUA2InputYiMajorityCond} of Algorithm~\ref{algm:OciorACOOL} are all eventually satisfied at the honest node~$i$.  
In this case, the honest node~$i$ eventually sets~$\ystrongmajority_i = y_i(\Meg_{1})$ and then sends~$\ltuple\NEWSYMBOL, \IDCOOL, \ystrongmajority_i\rtuple$ to all nodes (see Lines~\ref{line:BUA2InputYiMajorityCond} and~\ref{line:BUA2InputYiMajoritySend} of Algorithm~\ref{algm:OciorACOOL}).  
Therefore, in this case, the index of every honest node~$i$ that never sets $\RyNew_i^{[1]} =1$  is eventually included in~$\OECCorrectSymbolSetNew$ at each honest node, as in Line~\ref{line:BUA2InputYiMajorityyj} of Algorithm~\ref{algm:OciorACOOL}, such that 
\[
\OECCorrectSymbolSetNew[i] = y_i(\Meg_{1}), \quad \text{for} \  i \in [n] \setminus \Fc. 
\]

\vspace{8pt}
 \noindent \emph{$\bullet$ \underline{Proof of \eqref{eq:OECsymboli}  for Case~B:}}  \\
 We now consider Case~B, where an honest node~$i$ eventually sets $\RyNew_i^{[1]} = 1$.  
We assume that $|\AsetNOzero_{1}^{[2]}| \geq n - |\Fc| - t$ for~$\COOLBUA[1]$, and that all honest nodes eventually input their initial messages and keep running the~$\OciorACOOL$ and~$\COOLBUA[1]$ protocols.   
Under these assumptions, from Lemma~\ref{lm:OciorACOOLCaseIUnique2}, if an honest node~$i$ sets $\RyNew_i^{[1]} = 1$ and sends $\ltuple \SYMBOL, 1, (*, \ysymbolNew_i^{(i)}) \rtuple$ in~$\COOLBUA[1]$, then
\[
\ysymbolNew_i^{(i)} = y_i(\Meg_{1}).
\]
Here,~$\Meg_{1}$ is the initial message of all honest nodes within~$\AsetNOzero_{1}^{[2]}$.  
Note that under the assumptions considered here, every honest node~$i$ eventually sends $\ltuple \SYMBOL, 1, (*, \ysymbolNew_i^{(i)}) \rtuple$ in~$\COOLBUA[1]$.   
Therefore, under the same assumptions, every honest node~$i$ that sets $\RyNew_i^{[1]} = 1$ is eventually included in~$\OECCorrectSymbolSetNew$ at each honest node, as in Line~\ref{line:BUA2InputYiMajoritySi1} of Algorithm~\ref{algm:OciorACOOL} (or in Line~\ref{line:BUA2InputYiMajorityyj}, as in Case~A), such that
\[
\OECCorrectSymbolSetNew[i] = y_i(\Meg_{1}), \quad \text{for} \  i \in [n] \setminus \Fc. 
\]
This completes the proof.
\end{proof}

\begin{lemma}     \label{lm:OciorACOOLCaseIUnique}
Consider the case where $|\AsetNOzero_{1}^{[2]}| \geq n - |\Fc| - t$ for $\COOLBUA[1]$, and assume that all honest nodes eventually input their initial messages and keep running the $\OciorACOOL$ and $\COOLBUA[1]$ protocols. 
Under this assumption, if an honest node~$i$ sets $\ystrongmajority_i \gets y^{\star}$ in Line~\ref{line:BUA2InputYiMajoritySend} of Algorithm~\ref{algm:OciorACOOL}, then $y^{\star} = y_i(\Meg_{1})$.  
Here, $y_i(\Meg_{1})$   is the $i$-th coded symbol encoded from the message $\Meg_{1}$, while  $\Meg_{1}$ is the initial message of all honest nodes within $\AsetNOzero_{1}^{[2]}$.  
\end{lemma}
\begin{proof}
Let us now we consider the case where 
\begin{align}
|\AsetNOzero_{1}^{[2]}| \geq n - |\Fc| - t \label{eq:BUA2CaseI}
\end{align}
for $\COOLBUA[1]$,  and assume that all honest nodes eventually input their initial messages and keep running the $\OciorACOOL$ and $\COOLBUA[1]$ protocols. 
Here, $\AsetNOzero_{1}^{[2]} \defeqnew \{j \in [n] \setminus \Fc \mid \Me_j = \Meg_{1},\ \Ry_j^{[2]} \neq 0 \,\}$ for some $\Meg_{1}$. 
Under this assumption, if an honest node~$i$ sets $\ystrongmajority_i \gets y^{\star}$ in Line~\ref{line:BUA2InputYiMajoritySend} of Algorithm~\ref{algm:OciorACOOL}, then the conditions in Line~\ref{line:BUA2InputYiMajorityCond} must be satisfied, i.e.,
\begin{align}
|\ySymbolSet[y^{\star}] \cup \SsNew_0^{[2]}| \geq n - t, \quad \text{and} \quad |\ySymbolSet[y^{\star}]| \geq n - 2t 
\label{eq:BUA2InputYiMajoritySend}
\end{align}
 for some $y^{\star}$. 
Here, $\ySymbolSet[y^{\star}]$ includes the indices of nodes that sent $\ltuple \SYMBOL, 1, (\ysymbolNew_i^{(j)}, *) \rtuple$ to node~$i$ in $\COOLBUA[1]$ such that $\ysymbolNew_i^{(j)} = y^{\star}$ (see Lines~\ref{line:BUA2InputSymbolDeliverCond} and~\ref{line:BUA2InputSymbolDeliver} of Algorithm~\ref{algm:OciorACOOL}). 
The set $\SsNew_0^{[2]}$ includes the indices of nodes that sent $\Ry_j^{[2]} = 0$. 
Since each honest node~$j \in \AsetNOzero_{1}^{[2]}$ never sends $\Ry_j^{[2]} = 0$, it follows that
\begin{align}
\AsetNOzero_{1}^{[2]} \cap \SsNew_0^{[2]} = \emptyset.  \label{eq:ACOOLU0My}
\end{align} 
On the other hand,   it holds true that 
\begin{align}
| \AsetNOzero_{1}^{[2]} \cup \{\ySymbolSet[y^{\star}] \cup \SsNew_0^{[2]}\} | \leq  n .       \label{eq:ACOOLAllHonestbound}
\end{align}

Given the conditions in \eqref{eq:BUA2CaseI}-\eqref{eq:ACOOLAllHonestbound},   we now argue that $\AsetNOzero_{1}^{[2]} \cap \ySymbolSet[y^{\star}] \neq \emptyset$. 
We prove this result by contradiction and assume, for the sake of contradiction, that $\AsetNOzero_{1}^{[2]} \cap \ySymbolSet[y^{\star}] = \emptyset$. Specifically, under the assumption $\AsetNOzero_{1}^{[2]} \cap \ySymbolSet[y^{\star}] = \emptyset$ and given the condition $\AsetNOzero_{1}^{[2]} \cap \SsNew_0^{[2]} = \emptyset$ in \eqref{eq:ACOOLU0My},  we have the following bound: 
\begin{align}
| \AsetNOzero_{1}^{[2]} \cup \{\ySymbolSet[y^{\star}] \cup \SsNew_0^{[2]}\} |   &= \underbrace{| \AsetNOzero_{1}^{[2]} | }_{\geq n  -|\Fc| -t }  + \underbrace{|  \ySymbolSet[y^{\star}] \cup \SsNew_0^{[2]}| }_{ \geq n - t}       - \underbrace{| \AsetNOzero_{1}^{[2]} \cap \{\ySymbolSet[y^{\star}] \cup \SsNew_0^{[2]}\} | }_{=0}  \label{eq:ACOOLA1U099} \\
    &  \geq   (n  -|\Fc| -t) + (n - t )   \label{eq:ACOOLA1U088} \\
    &  \geq   n  -t -t + 3t+1 - t \label{eq:ACOOLA1U077} \\
    &  \geq    n     +1     \label{eq:ACOOLA1U066} 
\end{align} 
where \eqref{eq:ACOOLA1U088} follows from the inequalities in \eqref{eq:BUA2CaseI} and \eqref{eq:BUA2InputYiMajoritySend}, as well as the identity  
\[
\bigl|\AsetNOzero_{1}^{[2]} \cap \{\ySymbolSet[y^{\star}]\cup \SsNew_0^{[2]}\}\bigr| = 0
\]
under the assumption \(\AsetNOzero_{1}^{[2]} \cap \ySymbolSet[y^{\star}] = \emptyset\) and given the condition \(\AsetNOzero_{1}^{[2]} \cap \SsNew_0^{[2]} = \emptyset\) in \eqref{eq:ACOOLU0My}. Here \eqref{eq:ACOOLA1U077} uses the assumptions $n\geq 3t+1$ and $|\Fc|\leq t$. One can see that the conclusion in \eqref{eq:ACOOLA1U066} contradicts the result in \eqref{eq:ACOOLAllHonestbound}. Therefore, the original assumption \(\AsetNOzero_{1}^{[2]} \cap \ySymbolSet[y^{\star}] = \emptyset\) is false, and the statement $\AsetNOzero_{1}^{[2]} \cap \ySymbolSet[y^{\star}] \neq \emptyset$ holds true.
The conclusion $\AsetNOzero_{1}^{[2]} \cap \ySymbolSet[y^{\star}] \neq \emptyset$ implies that there exists at least one honest node \(j\in\AsetNOzero_{1}^{[2]}\) whose index also belongs to \(\ySymbolSet[y^{\star}]\). In other words, node \(j\) has sent  $\ltuple\SYMBOL, 1, (\ysymbolNew_i^{(j)},  *) \rtuple$ to node $i$ in $\COOLBUA[1]$ with $\ysymbolNew_i^{(j)}=y_i(\Meg_{1})=y^{\star}$ (see Lines~\ref{line:BUA2InputSymbolDeliverCond} and \ref{line:BUA2InputSymbolDeliver} of Algorithm~\ref{algm:OciorACOOL}). Consequently, if an honest node $i$ sets $\ystrongmajority_i\gets y^{\star}$ in Line~\ref{line:BUA2InputYiMajoritySend} of Algorithm~\ref{algm:OciorACOOL}, then $y^{\star}=y_i(\Meg_{1})$. 
\end{proof}

\begin{lemma}     \label{lm:OciorACOOLCaseIUnique2}
Consider the case where $|\AsetNOzero_{1}^{[2]}| \geq n - |\Fc| - t$ for $\COOLBUA[1]$, and assume that all honest nodes eventually input their initial messages and keep running the $\OciorACOOL$ and $\COOLBUA[1]$ protocols. 
Under this assumption,   if an honest node~$i$ sets $\Ry_i^{[1]} =1$ and sends $\ltuple \SYMBOL, 1, (*, \ysymbolNew_i^{(i)}) \rtuple$ in $\COOLBUA[1]$, then $\ysymbolNew_i^{(i)} = y_i(\Meg_{1})$.   Here,   $\Meg_{1}$ is the initial message of all honest nodes within $\AsetNOzero_{1}^{[2]}$.  
\end{lemma}
\begin{proof}
Similar to the proof of Lemma~\ref{lm:OciorACOOLCaseIUnique}, we also consider the case where
\begin{align}
|\AsetNOzero_{1}^{[2]}| \geq n - |\Fc| - t \label{eq:BUA2CaseINew}
\end{align}
for $\COOLBUA[1]$, and assume that all honest nodes eventually input their initial messages and continue executing the $\OciorACOOL$ and $\COOLBUA[1]$ protocols. 
Suppose further that there exists an honest node~$i$, where $i \in \Ac_{\Lin^\star}$ for some $\Lin^\star \in [\gn]$, such that node~$i$ sets $\Ry_j^{[1]} = 1$ and sends  $\ltuple \SYMBOL, 1, (*, \ysymbolNew_i^{(i)}) \rtuple$
in $\COOLBUA[1]$.  We now   argue that  $\ysymbolNew_i^{(i)} = y_i(\Meg_{1})$.
Specifically, when an honest node~$i \in \Ac_{\Lin^\star}$ sets $\Ry_i^{[1]} =1$, it implies that at least $n-t$  nodes in $\Lkset_1$ have sent  $\ltuple\SYMBOL, 1, (\ysymbolNew_i^{(j)},  *) \rtuple$ to node $i$ in $\COOLBUA[1]$ with 
\begin{align}
\ysymbolNew_i^{(j)}=\ysymbolNew_i^{(i)}, \quad  \forall j\in \Lkset_1     \label{eq:BUA2InputYibarji}
\end{align}
  (see Lines~\ref{line:ACOOLph1MatchCond} and \ref{line:COOLBUASISI1oneSend} of Algorithm~\ref{algm:COOLBUA}), where 
\begin{align}
 |\Lkset_1| \geq n-t  .   \label{eq:BUA2InputYiMajoritySendNew}
\end{align}
On the other hand,   it holds true that 
\begin{align}
| \AsetNOzero_{1}^{[2]} \cup \Lkset_1 | \leq  n .       \label{eq:ACOOLAllHonestboundNew}
\end{align} 

Given the conditions in \eqref{eq:BUA2CaseINew}-\eqref{eq:ACOOLAllHonestboundNew}, we argue that $\AsetNOzero_{1}^{[2]} \cap \Lkset_1 \neq \emptyset$.  This proof is similar to that of Lemma~\ref{lm:OciorACOOLCaseIUnique}. 
Specifically, if $\AsetNOzero_{1}^{[2]} \cap \Lkset_1 = \emptyset$, then we have 
\begin{align}
| \AsetNOzero_{1}^{[2]} \cup \Lkset_1 |  &= \underbrace{| \AsetNOzero_{1}^{[2]} | }_{\geq n  -|\Fc| -t }  + \underbrace{|  \Lkset_1| }_{ \geq n - t}       - \underbrace{| \AsetNOzero_{1}^{[2]} \cap \Lkset_1 | }_{=0}    \geq    n     +1     \label{eq:ACOOLA1U066New} 
\end{align}
which  contradicts the result in \eqref{eq:ACOOLAllHonestboundNew}. Therefore,   the statement $\AsetNOzero_{1}^{[2]} \cap \Lkset_1 \neq \emptyset$ holds true.
The conclusion $\AsetNOzero_{1}^{[2]} \cap \Lkset_1 \neq \emptyset$ implies that there exists at least one honest node $j\in\AsetNOzero_{1}^{[2]}$ whose index also belongs to $\Lkset_1$ such that the condition $\ysymbolNew_i^{(j)}=\ysymbolNew_i^{(i)}$ in \eqref{eq:BUA2InputYibarji} is satisfied. In other words, node~$j\in\AsetNOzero_{1}^{[2]}$ has sent  $\ltuple\SYMBOL, 1, (\ysymbolNew_i^{(j)},  *) \rtuple$ to node $i$ in $\COOLBUA[1]$ with $\ysymbolNew_i^{(j)}=y_i(\Meg_{1})=\ysymbolNew_i^{(i)}$. Consequently, in this case,  if an honest node~$i$ sets $\Ry_i^{[1]} =1$ and sends $\ltuple \SYMBOL, 1, (*, \ysymbolNew_i^{(i)}) \rtuple$ in $\COOLBUA[1]$, then $\ysymbolNew_i^{(i)} = y_i(\Meg_{1})$. 
\end{proof}

\begin{theorem}  [Communication, Round,  Resilience, and Computation]  \label{thm:OciorACOOLPerformance}
$\OciorACOOL$ achieves the $\ABA$ consensus  with total $O(\max\{n\Lh, n t \log \alphabetsize \})$  communication bits, $O(1)$  rounds, and a single invocation of a binary $\BA$ protocol,  under the optimal resilience assumption $n \geq 3t + 1$. Moreover, the computation complexity of $\OciorACOOL$ is $\tilde{O}(\Lh + t)$ in the good case and $\tilde{O}(t \Lh + t^2)$ in the worst case, measured in bit-level operations per node, where the good case occurs when the actual number of Byzantine nodes is small or the network is nearly synchronous.  
\end{theorem}
\begin{proof}
From Theorems~\ref{thm:OciorACOOLConsistency}-\ref{thm:OciorACOOLTermination}, it follows that the proposed $\OciorACOOL$ protocol satisfies the Consistency, Validity, and Termination properties in all executions, under the  assumption of an asynchronous binary $\BA$ protocol invoked within it, given that $n \geq 3t + 1$.

 Regarding the analysis of communication complexity, we first focus on the case where $t = \Omega(n)$. Recall that the size of the message to be agreed upon is $\ell$ bits. 
In $\OciorACOOL$, each coded symbol carries 
\[
\cb = \max\left\{ \frac{\ell}{k}, \, \log \alphabetsize \right\}
\]
bits, where $k = t/3$, and $\alphabetsize$ denotes the alphabet size of the error correction code used in the protocol.  
The total communication complexity of $\OciorACOOL$ is computed as
\begin{align}
\text{Total Comm.} 
  &= O( \cb n^2 + n^2 ) \nonumber\\
  &= O\!\left( \max\{ \ell n , \, n^2 \log \alphabetsize \} \right) \quad \text{bits}. 
  \label{eq:ComCase1}
\end{align}
In the above communication complexity analysis, we do not include the communication cost of the binary $\BA$ protocol, which is assumed to have a total communication complexity bounded by $O(n^2)$ bits and an expected $O(1)$ number of rounds.

Similar to $\COOL$, when $t$ is very small compared to $n$, we first select $n' \defeqnew 3t + 1$ nodes (e.g., the first $n'$ nodes), denoted by the set $\Sc'$, from the $n$ nodes to run the $\OciorACOOL$ protocol (see $\OciorACOOLsmallt$ protocol in Algorithm~\ref{algm:OciorACOOLsmallt}). 
In this setting, each coded symbol still carries $\cb$ bits. 
After the $\OciorACOOL$ protocol reaches agreement on a message $\Me'$ within $\Sc'$, the $i$-th node in $\Sc'$ sends a coded symbol $y_i'$, encoded from the agreed message $\Me'$, to the $n - n'$ nodes outside $\Sc'$, where   the symbols are encoded as 
\[
[y_1', y_2', \dotsc, y_{n'}'] \gets \ECCEnc(n', \kencode, \Me'). 
\] 
Each node outside $\Sc'$ eventually decodes the agreed message $\Me'$ using online error-correction decoding from the symbols $\{y_i'\}_{i \in \Sc'}$ collected from the nodes within $\Sc'$ (see Algorithm~\ref{algm:OciorACOOLsmallt}). 

Note that the $\OciorACOOL$ protocol executed within $\Sc'$ satisfies the Consistency, Validity, and Termination properties, given $n' = 3t + 1$. 
Furthermore, based on the decoding property of the error correction code, the overall protocol $\OciorACOOLsmallt$ (see Algorithm~\ref{algm:OciorACOOLsmallt})--comprising $\OciorACOOL$ within $\Sc'$ and the subsequent \emph{Strongly-Honest-Majority Distributed Multicast} ($\SHMDM$)  from $\Sc'$--also satisfies the Consistency, Validity, and Termination properties.  
In this case, the total communication complexity of the overall protocol $\OciorACOOLsmallt$ is computed as
\begin{align}
\text{Total Comm.} 
  &= O( \cb n' \cdot n'  + n'\cdot n' + \cb n' (n-n')  ) \nonumber\\
  &= O\!\left( \max\{\ell n, \, n t \log \alphabetsize\} \right) \quad \text{bits}.
  \label{eq:ComCase2}
\end{align}

Thus, by combining the result in~\eqref{eq:ComCase1} for the case $t = \Omega(n)$ and the result in~\eqref{eq:ComCase2} for the case of very small $t$, the total communication complexity of $\OciorACOOL$ can be expressed as 
\[
\text{Total Comm.} = O\!\left( \max\{ n\ell, \, n t \log \alphabetsize \} \right) \text{ bits.}
\]

 The proposed $\OciorACOOL$ protocol uses  $O(1)$   rounds and a single invocation of a binary $\BA$ protocol that is assumed to have  an expected $O(1)$ number of rounds.

The computation complexity of the proposed $\OciorACOOL$ is nominated by the online   error-correction  decoding in Lines~\ref{line:BUA2OECbegin}-\ref{line:BUA2OECend} and \ref{line:OECbegin}-\ref{line:OECend} of Algorithm~\ref{algm:OciorACOOL}. 
Fast $(n,k)$ Reed-Solomon decoding and polynomial error-correction decoding of degree-$(k-1)$ polynomials from $\nbar$ evaluation points (with $\nbar \le n$) are achievable in \[O(\nbar \log^2 \nbar \log \log \nbar)\] field operations using Fast Fourier Transform (FFT)-based polynomial arithmetic over the finite field $\Alphabet$ \cite{RS:60,Berlekamp:68, Gao:03, roth:06}.

If each symbol is represented using $\cb = \max\left\{ \frac{\ell}{k},  \log \alphabetsize \right\}$ bits, where $k=t/3$,  then  multiplication of two  symbols can be done in  $O(\cb \log\cb \log\log\cb)$ bit-level operations using fast multiplication algorithms,  while    inversion or   division can be done in  $O(\cb \log^2\cb \log\log\cb)$ bit-level operations using fast inversion algorithms. 
Therefore, fast $(n,k)$ Reed-Solomon decoding of the message $\Me$ can be achieved in 
\[
O((\nbar  \log^2 \nbar \log \log \nbar ) \cdot  (\cb \log^2\cb\log\log\cb) )  =\tilde{O}(\nbar     \cb )    =\tilde{O}(     \frac{\nbar \ell}{k}  + \nbar \log \alphabetsize)   
 = \tilde{O}(\ell + \nbar \log \alphabetsize )
\]
bit-level operations.  For the Reed-Solomon codes,  $\alphabetsize$ can be  set such that $\alphabetsize \geq  n+1$.   
In $\OciorACOOL$, for the error-correction decoding in Lines~\ref{line:BUA2OECbegin}-\ref{line:BUA2OECend} and \ref{line:OECbegin}-\ref{line:OECend} of Algorithm~\ref{algm:OciorACOOL}, $\nbar$ is set to $O(t)$, and the decoding may be repeated up to $t$ times.
Hence, the computational complexity of the proposed $\OciorACOOL$ in the worst case is
\[
t \cdot \tilde{O}(\ell + \nbar \log \alphabetsize )  =t \cdot \tilde{O}(\ell + t \log n )   =\tilde{O} ( t \ell +t^2 )
\]
bit-level operations per node.   The computation complexity of $\OciorACOOL$ in the good case  is \[\tilde{O}(\Lh + t)\]  bit-level operations per node, where the good case occurs when the actual number of Byzantine nodes is small or the network is nearly synchronous. 
\end{proof}

\subsection{Proofs of Lemma~\ref{lm:OciorACOOBUAUniqueSet}}    \label{sec:ProofBRCuniquegroup}
 
 In the following, we present Lemma~\ref{lm:OciorACOOBUAUniqueSet}, which is used in the above analysis. 
First, we provide several lemmas that will be used in the proof of Lemma~\ref{lm:OciorACOOBUAUniqueSet}. 
Lemma~\ref{lm:OciorACOOBUAUniqueSet} extends the result of~\cite[Lemma~3]{ChenOciorCOOL:24} (see Lemma~\ref{lm:OciorACOOL21} below).

 \begin{lemma} \cite[Lemma 7]{Chen:2020arxiv}   \label{lm:sizeboundMatrix} 
For $\gn  \geq 2$, it is true  that 
    \begin{align}
  |\Ac_{\Lin,j}| + |\Ac_{j,\Lin}|  <  &  k,    \quad \forall   j \neq \Lin,  \  j, \Lin \in [\gn]    
 \end{align} 
 where $k$ is a parameter of  $(n,k)$ error correction code.
  \end{lemma}

\begin{lemma}    \cite[Lemma 7]{ChenOciorCOOL:24}  \cite[Lemma 16]{ChenOciorCOOL:24}  \label{lm:OciorACOOL1221}
For  the   $\COOLBUA$  protocol with $n\geq 3t + 1$ and $k \le t/3$, if $\gn^{[1]} = 2$  then it holds true that $\gn^{[2]} \leq   1$.   
\end{lemma}

  \begin{lemma}    \cite[Lemma 3]{ChenOciorCOOL:24}  \cite[Lemma 11]{ChenOciorCOOL:24}  \label{lm:OciorACOOL21}
For   the  $\COOLBUA$  protocol with $n\geq 3t + 1$,  it holds true that $\gn^{[2]} \leq   1$.
\end{lemma}

\begin{lemma}    \label{lm:OciorACOOBUAUniqueSet}
For the $\COOLBUA$ protocol with $n \ge 3t + 1$ and $k \le t/3$, it holds true that $\gn^{[2]} \leq 1$. Furthermore, if all honest nodes eventually input their initial messages and keep running the $\COOLBUA$ protocol,  then it holds true that $\gnNozero^{[2]} \leq    1$, i.e., every honest node~$i$ eventually sets $\Ry_i^{[2]} = 0$ for all $i \in \Ac_{\Lin}$ and for all $\Lin \in [2, \gn]$.  
\end{lemma}
\begin{proof}   
 We first prove the first statement. For this statement, each honest node can terminate at any point in time, and some nodes may not have input their initial messages before termination. 
From Lemma~\ref{lm:OciorACOOLeta12New}, it follows that $\gn^{[1]} \leq 2$. 
Then, from Lemma~\ref{lm:OciorACOOL1221}, if $\gn^{[1]} = 2$, it holds that $\gn^{[2]} \leq 1$. 
If $\gn^{[1]} \leq 1$, it follows immediately that $\gn^{[2]} \leq 1$.

We now prove the second statement. 
From Lemma~\ref{lm:OciorACOOLeta12New}, if all honest nodes eventually input their initial messages and keep running the $\COOLBUA$ protocol, then $\gnNozero^{[1]} \leq 2$. 
Then, from Lemma~\ref{lm:OciorACOOL1221gnNozero}, if all honest nodes eventually input their initial messages and keep running the $\COOLBUA$ protocol, and if $\gnNozero^{[1]} = 2$, it follows that $\gnNozero^{[2]} \leq 1$. 
If $\gnNozero^{[1]} \leq 1$, it follows immediately that $\gnNozero^{[2]} \leq 1$. 
This completes the proof. 
 \end{proof}

\begin{lemma}    \label{lm:OciorACOOLeta12New}
For the $\COOLBUA$ protocol,  it is true that $\gn^{[1]} \leq 2$. Furthermore,  if  all honest nodes eventually  input their initial messages and keep running the $\COOLBUA$ protocol, then it  is also true that $\gnNozero^{[1]} \leq 2$.  
\end{lemma}
\begin{proof}
Recall that $\Lkset_1^{(i)}:=\{j\in [n] \mid   (y_i^{(j)}, y_j^{(j)}) = (y_i^{(i)}, y_j^{(i)})  \}$ and $\Lkset_0^{(i)}:=\{j\in [n] \mid   (y_i^{(j)}, y_j^{(j)}) \neq  (y_i^{(i)}, y_j^{(i)})  \}$ denote the  link indicator sets $\Lkset_1$ and $\Lkset_0$ updated by Node~$i$, as described in Line~\ref{line:ACOOLph1MatchCond} of Algorithm~\ref{algm:COOLBUA} (see \eqref{eq:LksetUseri}). 
Here we define $\LksetHonest_1^{(i)} := \Lkset_1^{(i)} \setminus \Fc$  and $\LksetHonest_0^{(i)} := \Lkset_0^{(i)} \setminus \Fc$ (see \eqref{eq:LksetHonestUseri}).     
   Recall that $\Ac_{\Lin} \defeqnew   \{  i \in[n]\setminus \Fc \mid    \Me_i =  \Meg_{\Lin}  \}$ for $\Lin\in [\gn]$ and for  some  non-empty $\ell$-bit distinct values $\Meg_{1}, \Meg_{2}, \cdots, \Meg_{\gn}$. 
Let us first prove the second statement and then the first of this lemma. 

\vspace{5pt}

 \noindent \emph{$\bullet$ \underline{Proof for    $\gnNozero^{[1]} \leq 2$, if   all honest nodes eventually  input     messages and keep running   $\COOLBUA$:}}  \\  
Let us first assume that    all honest nodes eventually  input their initial messages and    keep running the $\COOLBUA$ protocol, and that  there already exist an honest Node~$\iprime \in \Ac_{1}$ and an honest Node~$\iprimeprime \in \Ac_{2}$ such that they   \emph{never}  set $\Ry_{\iprime}^{[1]} = 0$ and $\Ry_{\iprimeprime}^{[1]} = 0$, respectively.  
We   now argue that each honest Node~$\ithree \in \Ac_{\Lin}$, for all $\Lin \in [3, \gn]$, will eventually set $\Ry_{\ithree}^{[1]} =0$. 

Given   $\Ry_{\iprime}^{[1]} \neq 0$ and $\Ry_{\iprimeprime}^{[1]} \neq 0$, we have the following bounds on $|\LksetHonest_0^{(\iprime)}| $ and $|\LksetHonest_0^{(\iprimeprime)}|$: 
\begin{align}
|\LksetHonest_0^{(\iprime)}|  \leq  t ,    \quad \text{and} \quad 
|\LksetHonest_0^{(\iprimeprime)}|  \leq   t.    \label{eq:Uinequalitylinkzero}
\end{align}
If  all honest nodes eventually  input their initial messages and keep running  the $\COOLBUA$ protocol, then for each $i\in [n]\setminus \Fc$,   eventually the following equality holds true: 
\begin{align}
|\LksetHonest_1^{(i)}|  + |\LksetHonest_0^{(i)}|    = n-|\Fc| .   \label{eq:VVbound11} 
\end{align}
Then, the above results in  \eqref{eq:Uinequalitylinkzero} and  \eqref{eq:VVbound11} imply:    
\begin{align}
|\LksetHonest_1^{(\iprime)}|  \geq  n - t - |\Fc|,    \quad \text{and} \quad 
|\LksetHonest_1^{(\iprimeprime)}|  \geq  n - t - |\Fc|.    \label{eq:Uinequality222}
\end{align}
Furthermore, for each $j \in \LksetHonest_1^{(\iprime)} \cap \LksetHonest_1^{(\iprimeprime)}$, it is true that $\hv_j^\T \Meg_{1} = \hv_j^\T \Me_j$ and $\hv_j^\T \Meg_{2} = \hv_j^\T \Me_j$ (see Line~\ref{line:ACOOLph1MatchCond} of Algorithm~\ref{algm:COOLBUA}), which implies that   $\hv_j^\T \Meg_{1} = \hv_j^\T \Meg_{2}$,   $\forall j \in \LksetHonest_1^{(\iprime)} \cap \LksetHonest_1^{(\iprimeprime)}$.  
Thus,    we conclude:  
\begin{align}
|\LksetHonest_1^{(\iprime)} \cap \LksetHonest_1^{(\iprimeprime)}| \leq k - 1,   \label{eq:VVproperty1}   
\end{align}
otherwise $\Meg_{1} = \Meg_{2}$ due to a property of linear algebra, which contradicts our assumption  $\Meg_{1} \neq \Meg_{2}$. 
Note that if there exists a full-rank matrix $H$ of size  $k\times k$ such that $H \xv =\boldsymbol{0}$, then it follows  that $\xv=\boldsymbol{0}$.

Similarly, we have 
\begin{align}
|\LksetHonest_1^{(\iprime)} \cap \LksetHonest_1^{(\ithree)}| \leq k - 1, \quad \text{and} \quad 
|\LksetHonest_1^{(\iprimeprime)} \cap \LksetHonest_1^{(\ithree)}| \leq k - 1.   \label{eq:VVproperty3}   
\end{align}

In our setting,   the following identities hold true: 
\begin{align} 
|\LksetHonest_1^{(\iprime)} \cup\LksetHonest_1^{(\iprimeprime)}\cup  \LksetHonest_1^{(\ithree)}| \leq  n-|\Fc|,   \label{eq:VVbound1}   
\end{align}
\begin{align} 
|\LksetHonest_1^{(\iprime)} \cup\LksetHonest_1^{(\iprimeprime)}\cup  \LksetHonest_1^{(\ithree)}| \geq |\LksetHonest_1^{(\iprime)}|\!+ |\LksetHonest_2^{(\iprimeprime)}| \!+  |\LksetHonest_1^{(\ithree)}|\!- |\LksetHonest_1^{(\iprime)} \cap \LksetHonest_1^{(\iprimeprime)}| \!-|\LksetHonest_1^{(\iprime)} \cap \LksetHonest_1^{(\ithree)}|\!- |\LksetHonest_1^{(\iprimeprime)} \cap \LksetHonest_1^{(\ithree)}| ,   \label{eq:VVbound2}   
\end{align}
where the first identity follows from the fact $\LksetHonest_1^{(\iprime)} \cup\LksetHonest_1^{(\iprimeprime)}\cup  \LksetHonest_1^{(\ithree)} \subseteq [n] \setminus \Fc$ and the  second identity follows from the inclusion-exclusion identity for the union of three sets.  At this point,  $|\LksetHonest_1^{(\ithree)}|$ can be bounded by: 
\begin{align}
|\LksetHonest_1^{(\ithree)}|  &\leq  |\LksetHonest_1^{(\iprime)} \cup\LksetHonest_1^{(\iprimeprime)}\cup  \LksetHonest_1^{(\ithree)}| -  |\LksetHonest_1^{(\iprime)}|\!-  |\LksetHonest_2^{(\iprimeprime)}| + |\LksetHonest_1^{(\iprime)} \cap \LksetHonest_1^{(\iprimeprime)}| \!+|\LksetHonest_1^{(\iprime)} \cap \LksetHonest_1^{(\ithree)}|\!+ |\LksetHonest_1^{(\iprimeprime)} \cap \LksetHonest_1^{(\ithree)}|    \label{eq:VVbound3}    \\
  &\leq  n-|\Fc|   -2(n - t - |\Fc|)   +3(k-1) \label{eq:VVbound4} \\
    &\leq  n-|\Fc|   -2\cdot ( 3t+1 - t - t)   +3\cdot(t/3-1) \label{eq:VVbound5} \\
        &=   n-|\Fc|   - t  -5 \label{eq:VVbound6} \\
    &<  n-|\Fc|   - t  \label{eq:VVbound7} 
\end{align}
where the first inequality follows from  the identity in \eqref{eq:VVbound2};     the second inequality  follows from the results in \eqref{eq:Uinequality222}-\eqref{eq:VVbound1};   and the third inequality uses  the identities   $n\geq 3t+1$,  $|\Fc| \leq  t$ and $k \leq t/3$.

Since $ \LksetHonest_0^{(i)} \subseteq \Lkset_0^{(i)}$, and from \eqref{eq:VVbound11} and \eqref{eq:VVbound6}, if  all honest nodes eventually  input their initial messages and keep running  the $\COOLBUA$ protocol,   then we have 
\begin{align}
|\Lkset_0^{(\ithree)}|  &\geq |\LksetHonest_0^{(\ithree)}|     \non \\
    &= n-|\Fc| -   |\LksetHonest_1^{(\ithree)}|  \label{eq:VVbound9} \\
    & \geq   n-|\Fc|   - ( n-|\Fc|   - t -5) \label{eq:VVbound610} \\
        & >     t +1 \label{eq:VVbound77} 
\end{align}
where \eqref{eq:VVbound9} follows from \eqref{eq:VVbound11}; and   \eqref{eq:VVbound610} is from \eqref{eq:VVbound6}. 
The result in \eqref{eq:VVbound77} reveals that, if   all honest nodes eventually  input their initial messages and keep running  the $\COOLBUA$ protocol,  then each honest Node~$\ithree \in \Ac_{\Lin}$, for all $\Lin \in [3, \gn]$, will eventually set $\Ry_{\ithree}^{[1]} = 0$.  Thus,  it is true that  $\gnNozero^{[1]} \leq 2$.

 \vspace{5pt}

 \noindent \emph{$\bullet$ \underline{Proof for    $\gn^{[1]} \leq 2$:}}  \\  
 Let us assume that there already exist an honest node~$\iprime \in \Ac_{1}$ and an honest node~$\iprimeprime \in \Ac_{2}$ such that they have set $\Ry_{\iprime}^{[1]} = 1$ and $\Ry_{\iprimeprime}^{[1]} = 1$, respectively, in  $\COOLBUA$ protocol. In this case, we assume that each honest node can terminate at any point in time, and that some nodes may not have input their initial messages before termination. 
We now argue that every honest node~$\ithree \in \Ac_{\Lin}$, for all $\Lin \in [3, \gn]$, will never set $\Ry_{\ithree}^{[1]} = 1$.

Under the   assumption of $\Ry_{\iprime}^{[1]} = 1$ and $\Ry_{\iprimeprime}^{[1]} = 1$, and from the condition in Line~\ref{line:COOLBUASISI1oneCond} of Algorithm~\ref{algm:COOLBUA}, the following inequities hold true: 
\begin{align}
|\Lkset_1^{(\iprime)}|   \geq  n-t,    \quad  \text{and} \quad 
|\Lkset_1^{(\iprimeprime)}|  \geq  n-t. \label{eq:Uinequality2}
\end{align} 
Then, the above result implies:    
\begin{align}
|\LksetHonest_1^{(\iprime)}|  \geq  n - t - |\Fc|,    \quad \text{and} \quad 
|\LksetHonest_1^{(\iprimeprime)}|  \geq  n - t - |\Fc|.    \label{eq:Uinequality22}
\end{align} 
One can check that, under this assumption that   $\Ry_{\iprime}^{[1]} = 1$ and $\Ry_{\iprimeprime}^{[1]} = 1$,  the results in  \eqref{eq:Uinequality222}-\eqref{eq:VVbound7} hold true 
and it is concluded that 
\begin{align}
|\LksetHonest_1^{(\ithree)}|   <  n-|\Fc|   - t  \label{eq:VVbound777} 
\end{align}
(see \eqref{eq:VVbound7}). 
This result   reveals that each honest Node~$\ithree \in \Ac_{\Lin}$, for all $\Lin \in [3, \gn]$, will never set $\Ry_{\ithree}^{[1]} = 1$.  Thus,  it is true that  $\gn^{[1]} \leq 2$.  This completes the proof.  
\end{proof}

\begin{lemma}    \label{lm:OciorACOOL1221gnNozero}
For  the   $\COOLBUA$  protocol with $n\geq 3t + 1$ and $k \le t/3$,   and assuming that  all honest nodes eventually input their initial messages and keep running the $\COOLBUA$ protocol,   if $\gnNozero^{[1]} = 2$  then it holds true that $\gnNozero^{[2]} \leq   1$.   
\end{lemma}
\begin{proof}
This proof  follows closely  that of  \cite[Lemma 7]{ChenOciorCOOL:24}  \cite[Lemma 16]{ChenOciorCOOL:24}.    Here we assume that  all honest nodes eventually input their initial messages and keep running the $\COOLBUA$ protocol.  We also assume that $\gnNozero^{[1]} = 2$. 
Under this assumption,  the definition in \eqref{eq:ACOOLAell00}-\eqref{eq:ACOOLAll22} suggests that 
 \begin{align} 
\AsetNOzero_{1}^{[1]} \defeqnew&   \{  i \in[n]\setminus \Fc \mid     \Me_i =  \Meg_{1},   \  \Ry_i^{[1]} \neq 0 \},    \label{eq:ACOOLANoZeroNew1}   \\
\AsetNOzero_{2}^{[1]} \defeqnew&   \{  i \in[n]\setminus \Fc \mid     \Me_i =  \Meg_{2},   \  \Ry_i^{[1]} \neq 0 \},   \label{eq:ACOOLANoZeroNew2}   \\
\Bc^{[1]} \defeqnew  & \{  i \in[n]\setminus \Fc \mid  \Ry_i^{[1]} =0   \} \label{eq:ACOOLBdef01New2} \\
  \AsetNOzero_{2,1}^{[1]} \defeqnew   & \{  i\in  \AsetNOzero_{2}^{[1]} \mid  \hv_i^\T  \Meg_{2}  = \hv_i^\T  \Meg_1 \}   \label{eq:Alj22New21}  \\     
  \AsetNOzero_{2,2}^{[1]} \defeqnew  &  \AsetNOzero_{2}^{[1]}   \setminus \AsetNOzero_{2,1}^{[1]} =  \{  i\in  \AsetNOzero_{2}^{[1]} \mid  \hv_i^\T  \Meg_{2}  \neq  \hv_i^\T  \Meg_1 \}   \label{eq:Alj22New22} \\
    \AsetNOzero_{1,2}^{[1]} \defeqnew   & \{  i\in  \AsetNOzero_{1}^{[1]} \mid  \hv_i^\T  \Meg_{1}  = \hv_i^\T  \Meg_2 \}   \label{eq:Alj22New12}  \\     
  \AsetNOzero_{1,1}^{[1]} \defeqnew  &  \AsetNOzero_{1}^{[1]}   \setminus \AsetNOzero_{1,2}^{[1]} =  \{  i\in  \AsetNOzero_{1}^{[1]} \mid  \hv_i^\T  \Meg_{1}  \neq  \hv_i^\T  \Meg_2 \}   \label{eq:Alj22New11} 
 \end{align}  
Given the identity    $|\Ac_{1}|+|\Ac_{2}|= n-|\Fc|-\sum_{\Lin=3}^{\gn} |\Ac_{\Lin}|$,  we will consider each of the following two cases: 
\begin{align}
\text{Case~1:} \quad  
|\Ac_{2}| &\leq    \frac{n-|\Fc|-\sum_{\Lin=3}^{\gn} |\Ac_{\Lin}|}{2} ,  \label{eq:ACOOLcase1b} \\
\text{Case~2:} \quad |\Ac_{1}| &\leq  \frac{n-|\Fc|-\sum_{\Lin=3}^{\gn} |\Ac_{\Lin}|}{2}  .   \label{eq:ACOOLcase2a} 
 \end{align}

  \noindent \emph{$\bullet$ \underline{Analysis for Case~1:}}  \\   
We   first consider Case~1.  
Under the assumption that  all honest nodes eventually input their initial messages and keep running the $\COOLBUA$ protocol, for each Node~$i \in \AsetNOzero_{2}^{[1]}$,  it is true that 
 \begin{align}  
\Bc^{[1]}  \in  \Ss_0^{[1]}\cup \Lkset_0 ,        \label{eq:B1inZeroSet} 
\end{align}  
 \begin{align}  
  \AsetNOzero_{1,1}^{[1]}  \in  \Ss_0^{[1]}\cup \Lkset_0   ,      \label{eq:A11inZeroSet} 
\end{align}  
and that 
 \begin{align}  
 |\Ss_0^{[1]}\cup \Lkset_0| &\geq   |\Bc^{[1]} \cup    \AsetNOzero_{1,1}^{[1]}  |  \non\\
& =  n-|\Fc|  - ( |\AsetNOzero_{1,2}^{[1]}|+ |\AsetNOzero_{2}^{[1]}|)\non\\
   & =  n-\underbrace{|\Fc|}_{\leq t}  - (\underbrace{ |\AsetNOzero_{1,2}^{[1]}|  +  |\AsetNOzero_{2,1}^{[1]}|}_{\leq k-1}+  \underbrace{|\AsetNOzero_{2,2}^{[1]}|}_{\leq 2(k-1)})\non\\
 & \geq  n-t-3(k-1)           \label{eq:Ss0Lkset0Bound88} \\
   & >   t+1.            \label{eq:Ss0Lkset0Bound99} 
\end{align}
 Here, \eqref{eq:Ss0Lkset0Bound88} follows from Lemma~\ref{lm:ACOOLboundA22} and the identity that $|\AsetNOzero_{1,2}^{[1]}|  +  |\AsetNOzero_{2,1}^{[1]}| \leq |\Ac_{1,2}|  +  |\Ac_{2,1}| \leq k-1$  (see Lemma~\ref{lm:sizeboundMatrix}), while  \eqref{eq:Ss0Lkset0Bound99} uses the assumption of $n\geq 3t+1$ and $k\leq t/3$.  
The result $|\Ss_0^{[1]}\cup \Lkset_0| > t+1$ in \eqref{eq:Ss0Lkset0Bound99} reveals that  each Node~$i \in \AsetNOzero_{2}^{[1]}$ eventually sets $\Ry_i^{[2]} = 0$, as described in Lines~\ref{line:ACOOLph2CSI0Cond} and \ref{line:ACOOLph2CSI0} of Algorithm~\ref{algm:COOLBUA},  for Case~1.

  \noindent \emph{$\bullet$ \underline{Analysis for Case~2:}}  \\   
  By interchanging the roles of $\Ac_{1}$ and $\Ac_{2}$ and following the proof for Case~1, one can show that each Node $i \in \AsetNOzero_{1}^{[1]}$ eventually sets $\Ry_i^{[2]} = 0$, as described in Line~\ref{line:ACOOLph2CSI0} of Algorithm~\ref{algm:COOLBUA}, for Case~2.
\end{proof}

\begin{lemma}    \label{lm:ACOOLboundA22}
Given   $|\Ac_{i}| \leq    \frac{n-|\Fc|-\sum_{\Lin=3}^{\gn} |\Ac_{\Lin}|}{2}$ and   $\gn\geq 2$, and assuming that  all honest nodes eventually input their initial messages and keep running the $\COOLBUA$ protocol,  it is true that 
\begin{align}
|\AsetNOzero_{i,i}^{[1]}| \leq  2(k-1)    
\end{align}
for $i\in \{1,2\}$.    
\end{lemma}
\begin{proof}
This proof closely follows  that of  \cite[Lemma 8]{ChenOciorCOOL:24}.  
Here we assume that  all honest nodes eventually input their initial messages and keep running the $\COOLBUA$ protocol.  
Without loss of generality, we  just  focus on the proof of  $|\AsetNOzero_{2,2}^{[1]}| \leq  2(k-1)$, given  the condition $|\Ac_{2}| \leq    \frac{n-|\Fc|-\sum_{\Lin=3}^{\gn} |\Ac_{\Lin}|}{2}$ and   $\gn\geq 2$.  

We first  consider the case where $\gn \geq 3$.  
Recall that $\Lkset_1^{(i)}:=\{j\in [n] \mid   (y_i^{(j)}, y_j^{(j)}) = (y_i^{(i)}, y_j^{(i)})  \}$ and $\Lkset_0^{(i)}:=\{j\in [n] \mid   (y_i^{(j)}, y_j^{(j)}) \neq  (y_i^{(i)}, y_j^{(i)})  \}$ denote the  link indicator sets $\Lkset_1$ and $\Lkset_0$ updated by Node~$i$, as described in Line~\ref{line:ACOOLph1MatchCond}  of Algorithm~\ref{algm:COOLBUA} (see \eqref{eq:LksetUseri}). 
Also recall that $\LksetHonest_1^{(i)} := \Lkset_1^{(i)} \setminus \Fc$  and $\LksetHonest_0^{(i)} := \Lkset_0^{(i)} \setminus \Fc$.    
Here $\Lk_i(j) \in \{0,1\}$   denotes the link indicator between Node~$i$ and Node~$j$, defined in \eqref{eq:lkindicator}.  
Under the assumption  that all honest nodes eventually input their initial messages and keep running the $\COOLBUA$ protocol, since each Node~$i \in \AsetNOzero_{2,2}^{[1]}$  never set   $ \Ry_i^{[1]} = 0$,  it holds true that  
 \begin{align}
|\LksetHonest_0^{(i)}| \leq  t,\quad     \forall i\in \AsetNOzero_{2,2}^{[1]} .    \label{eq:A22conditionNozero99} 
\end{align}
If   all honest nodes eventually input their initial messages and keep running the $\COOLBUA$ protocol, it is true that $|\LksetHonest_1^{(i)}|  + |\LksetHonest_0^{(i)}|    = n-|\Fc|$. 
Thus, from the above identity and from \eqref{eq:A22conditionNozero99},   we have 
 \begin{align}
|\LksetHonest_1^{(i)}| =n-|\Fc| -  |\LksetHonest_0^{(i)}|  \geq  n-|\Fc|  -  t,\quad     \forall i\in \AsetNOzero_{2,2}^{[1]} .    \label{eq:A22conditionNozero88} 
\end{align}
The result in \eqref{eq:A22conditionNozero88} implies that 
 \begin{align}
\sum_{ j \in [n]\setminus \Fc}   \Lk_{i}  (j)  \geq  n- |\Fc| - t,\quad    \forall i\in \AsetNOzero_{2,2}^{[1]} .     \label{eq:A22conditionNozero77} 
\end{align}
and   that 
\begin{align}
\sum_{ j \in [n]\setminus (\Fc\cup \Ac_{2})}   \Lk_{i}  (j)  \geq n-t -|\Fc|-|\Ac_{2}|,\quad      \forall i\in \AsetNOzero_{2,2}^{[1]} .     \label{eq:A22conditionNozero66} 
\end{align}
In this setting, since $\hv_i^\T  \Meg_{1}  \neq  \hv_i^\T  \Meg_2 $ for any  $i\in \AsetNOzero_{2,2}^{[1]}$  (see \eqref{eq:Alj22New22}),   we conclude that $\Lk_{i} (j) =0$,  $\forall i\in \AsetNOzero_{2,2}^{[1]}, \forall j \in \Ac_{1}$. 
This result implies that   the inequality in \eqref{eq:A22conditionNozero66} can be updated as 
 \begin{align}
\sum_{ j \in [n]\setminus (\Fc\cup \Ac_{2}\cup \Ac_{1})}   \Lk_{i} (j)  \geq n-t -|\Fc|-|\Ac_{2}|,\quad     \forall i\in \AsetNOzero_{2,2}^{[1]} .     \label{eq:A22conditionNozero55} 
\end{align}
Given    $[n]\setminus (\Fc\cup \Ac_{2}\cup \Ac_{1}) =\cup_{\Lin=3}^{\gn}\Ac_{\Lin}$,  the result in \eqref{eq:A22conditionNozero55} can be rewritten as  
 \begin{align}
\sum_{ j \in\cup_{\Lin=3}^{\gn}\Ac_{\Lin}}   \Lk_{i} (j)  \geq n-t -|\Fc|-|\Ac_{2}|,\quad     \forall i\in \AsetNOzero_{2,2}^{[1]}   \label{eq:A22conditionNozero44} 
\end{align}
which implies  the following bound 
 \begin{align}
\sum_{ i \in \AsetNOzero_{2,2}^{[1]}}  \sum_{ j \in\cup_{\Lin=3}^{\gn}\Ac_{\Lin}}   \Lk_{i} (j)    \geq (n-t -|\Fc|-|\Ac_{2}|) \cdot |\AsetNOzero_{2,2}^{[1]} | .  \label{eq:A22conditionNozero91875} 
\end{align}

On the other hand,   for  any given $j\in   \Ac_{ \Lin^{\star}}$ and  $ \Lin^{\star} \in [3,\gn]$,     we have   
\begin{align}
\sum_{ i \in \AsetNOzero_{2,2}^{[1]}}   \Lk_{i} (j)  \leq\!  \! \sum_{ i \in \Ac_{2}^{[1]}}   \Lk_{i} (j)  =   \! \sum_{ i \in \Ac_{2,   \Lin^{\star}}^{[1]}}   \Lk_{i} (j)   +\!\!\! \sum_{ i \in \Ac_{2}^{[1]}  \setminus \Ac_{2,   \Lin^{\star}}^{[1]}}   \Lk_{i} (j)      
=  \!\! \! \sum_{ i \in \Ac_{2,   \Lin^{\star}}^{[1]}}   \Lk_{i} (j)    
 \leq     |\Ac_{2,   \Lin^{\star}}^{[1]}|    
 \leq     k-1     \label{eq:A22conditionNozero81857}
\end{align}
where the above result uses the identity that $\hv_i^\T  \Meg_{ \Lin^{\star}}\neq   \hv_i^\T  \Meg_{2}$ for $i \in \Ac_{2}^{[1]}  \setminus \Ac_{2,   \Lin^{\star}}^{[1]}$ (see  \eqref{eq:Alj11}),  and  from Lemma~\ref{lm:sizeboundMatrix}. 
The result in \eqref{eq:A22conditionNozero81857} then implies that  
 \begin{align}
\sum_{ j \in\cup_{\Lin=3}^{\gn}\Ac_{\Lin}} \sum_{ i \in \AsetNOzero_{2,2}^{[1]}}   \Lk_{i} (j)  \leq    (k-1) \cdot \sum_{\Lin=3}^{\gn} |\Ac_{\Lin}|  \label{eq:A22conditionNozero43536} 
\end{align}

From    the inequalities     \eqref{eq:A22conditionNozero91875} and  \eqref{eq:A22conditionNozero43536},  we have $(n-t -|\Fc|-|\Ac_{2}|) \cdot |\AsetNOzero_{2,2}^{[1]} |   \leq    (k-1) \cdot \sum_{\Lin=3}^{\gn} |\Ac_{\Lin}| $ and 
 \begin{align}
|\AsetNOzero_{2,2}^{[1]}| \leq    \frac{(k-1) \cdot \sum_{\Lin=3}^{\gn} |\Ac_{\Lin}|}{n-t- |\Fc| -|\Ac_{2}| }  ,    \label{eq:A22conditionNozero32535}  
\end{align}
where $n-t- |\Fc| -|\Ac_{2}|>0$ is  true  given $|\Ac_{2}| \leq    \frac{n-|\Fc|-\sum_{\Lin=3}^{\gn} |\Ac_{\Lin}|}{2}$  and $n\geq 3t+1$. 
The bound in    \eqref{eq:A22conditionNozero32535} can be further extended as  
\begin{align}
|\AsetNOzero_{2,2}^{[1]}|  
\leq    \frac{(k-1) \cdot \sum_{\Lin=3}^{\gn} |\Ac_{\Lin}|}{n-t- |\Fc| - \frac{n-|\Fc|-\sum_{\Lin=3}^{\gn} |\Ac_{\Lin}|}{2}}    
=    \frac{2(k-1)  }{\frac{n-2t- |\Fc|}{\sum_{\Lin=3}^{\gn} |\Ac_{\Lin}|} + 1}    
\leq     \frac{2(k-1)  }{0 + 1}      
=    2(k-1)       \label{eq:A22bound82755}  
\end{align}
where   the first inequality follows from  the condition $|\Ac_{2}| \leq    \frac{n-|\Fc|-\sum_{\Lin=3}^{\gn} |\Ac_{\Lin}|}{2}$;  and the second inequality  results from  the fact  that   $\frac{n-2t- |\Fc|}{\sum_{\Lin=3}^{\gn} |\Ac_{\Lin}|} > 0$ in this   case with $\gn \geq 3$.  

We now focus on  the case where $\gn = 2$.    
We assume that $|\AsetNOzero_{2,2}^{[1]}| >0$.  Under this assumption,  by following the steps in \eqref{eq:A22conditionNozero99}-\eqref{eq:A22conditionNozero44}, and given $\gn = 2$, we have 
 \begin{align}
0=\sum_{ j \in\cup_{\Lin=3}^{\gn}\Ac_{\Lin}}   \Lk_{i} (j)  \geq n-t -|\Fc|-|\Ac_{2}|,\quad     \forall i\in \AsetNOzero_{2,2}^{[1]} .   \label{eq:A22condition28366BB} 
\end{align}
The bound in  \eqref{eq:A22condition28366BB}   contradicts the condition   $|\Ac_{2}| \leq    \frac{n-|\Fc|-\sum_{\Lin=3}^{\gn} |\Ac_{\Lin}|}{2} < n-t -|\Fc|$. Hence, the assumption   $|\AsetNOzero_{2,2}^{[1]}| >0$ leads to a contradiction; therefore,  it is true that $|\AsetNOzero_{2,2}^{[1]}|=0$   for this case with  $\gn = 2$.
This   completes the proof. 
\end{proof}

\begin{algorithm}
\caption{$\OciorACOOLsmallt$ protocol with identifier $\IDCOOL$ for a small $t$.  Code is shown for    Node~$\thisnodeindex \in [n]$. }  \label{algm:OciorACOOLsmallt}
\begin{algorithmic}[1]

\footnotesize
 
  \Statex   \emph{//   ** This protocol is designed for the case where $t$ is very small compared to $n$ **}

\Statex \emph{// ** It includes  $\OciorACOOL$ within   $\Sc' := [n']$, 
and  a \emph{Strongly-Honest-Majority Distributed Multicast} ($\SHMDM$) from $\Sc'$.  **}

\Statex \emph{// ** Here $n' \defeqnew 3t + 1$. **}

\Statex
\State  Initially set $n' \gets  3t + 1; \kencode \gets      \networkfaultsizet / 3; \OECCorrectSymbolSet\gets \{\}$ 
 
 \Statex

  \Statex   \emph{//   ********************  $\OciorACOOL$********************}

\State {\bf upon} receiving a non-empty  message  input $\Me_{i}$  and $i\in [n']$ {\bf do}:

\Indent  

 \State    $\Pass$ $\Me_{i}$ into  $\OciorACOOL$    as an input value
 \State    run the $\OciorACOOL$ protocol with other nodes within $[n']$ 
\EndIndent

 \Statex
    			
     \Statex   \emph{//   ********************   $\SHMDM$ ********************}  
      			
 \State {\bf upon} outputting a message    $\Me'$  from  $\OciorACOOL$ protocol,   and $i\in [n']$   {\bf do}:       
\Indent  
	\State $[y_1', y_2', \dotsc, y_{n'}'] \gets \ECCEnc(n', \kencode, \Me')$
 	\State $\send$  $\ltuple\SHMDM, \IDCOOL, y_i \rtuple$  to  all nodes within $[n]\setminus [n']$ 	     
	\State $\Output$   $\Me'$ and $\terminate$     		     
\EndIndent

\State {\bf upon} receiving   $\ltuple\SHMDM, \IDCOOL, y_j' \rtuple$  from  Node~$j \in [n']$ for the first time, and  $i\notin [n']$   {\bf do}:   
\Indent  
	\State $\OECCorrectSymbolSet[j] \gets y_j' $    
 	\If  { $|\OECCorrectSymbolSet|\geq  \kencode + \networkfaultsizet  $}      \quad    \quad\quad \quad \quad\quad\quad \quad    \quad\quad \quad \quad\quad\quad \quad   \quad\quad \quad \quad\quad\quad \quad   \quad\quad \quad \quad\quad\quad \quad    \emph{//   online error correcting   }  
			\State   $\MVBAOutputMsg  \gets \ECCDec(n',  \kencode , \OECCorrectSymbolSet)$	     
			\State  $[y_{1}, y_{2}, \cdots, y_{n'}] \gets \ECCEnc (n',  \kencode, \MVBAOutputMsg)$ 
    			\If {at least $\kencode + \networkfaultsizet$ symbols in $[y_{1}, y_{2}, \cdots, y_{n'}]$ match with  those in $\OECCorrectSymbolSet$}
				\State $\Output$   $\MVBAOutputMsg$ and $\terminate$     		     
    			\EndIf    

	\EndIf   
		     
\EndIndent

\end{algorithmic}
\end{algorithm}

\begin{algorithm}
\caption{$\OciorRBA$ protocol with identifier $ \IDCOOL $.  Code is shown for    Node~$\thisnodeindex$ for $ \thisnodeindex \in [n]$. }  \label{algm:OciorRBA}   
\begin{algorithmic}[1]

\footnotesize
 
\State  Initially set   $\kencode \gets      \networkfaultsizet / 3;  \Me^{(i)}\gets \defaultvalue;  \OECSIFinal\gets 0;    \OECCorrectSymbolSet\gets \{\};     \Phthreeindicator\gets 0 $

 \vspace{3pt}

  \Statex   \emph{//   ********************  $\COOLBUA$********************} 
	
\State {\bf upon} receiving a non-empty  message  input $\Me_{i}$ {\bf do}:

\Indent  

 \State    $\Pass$ $\Me_{i}$ into  $\COOLBUA$    as an input value
 
\EndIndent

 \vspace{3pt}

  \Statex   \emph{//   ******************** Binary  $\RBA$  ($\BRBA$)********************}

\State {\bf upon}  delivery of $\Ss_{\Vr^\star}^{[2]}$  from   $\COOLBUA$ such that  $|\Ss_{\Vr^\star}^{[2]}|\geq  n- t$,     for a  $\Vr^\star \in \{1,0\}$,  and $\ltuple\READY, \IDCOOL, * \rtuple$  not yet sent {\bf do}:     
\Indent  
		\State $\send$ $\ltuple\READY, \IDCOOL, \Vr^\star \rtuple$ to  all nodes  	\    \emph{// For some application in \cite{ChenOciorMVBA:24},  the value of  $\Vr^\star$ is  delivered   to the protocol invoking $\OciorRBA$}       
\EndIndent

\State {\bf upon} receiving   $\networkfaultsizet+1$  $\ltuple \READY, \IDCOOL, \Vr  \rtuple$ messages  from different   nodes for the same $\Vr$, and $\ltuple\READY, \IDCOOL, * \rtuple$  not yet sent {\bf do}:          
\Indent  
		\State $\send$ $\ltuple\READY, \IDCOOL, \Vr \rtuple$ to  all nodes	   
\EndIndent

\State {\bf upon} receiving   $2t+1$  $\ltuple \READY, \IDCOOL, \Vr  \rtuple$ messages  from different nodes   for the same $\Vr$ {\bf do}:    
\Indent  
 	\State  set $\VrOutput\gets \Vr$            
     \IfThenElse {$\VrOutput =0$}  {$\Output$   $\Me^{(i)}=\defaultvalue$ and $\terminate$ } {set $\Phthreeindicator\gets 1$ }     
     
\EndIndent

 \vspace{3pt}
  
  \Statex   \emph{//   ********************  $\COOLHMDM$ ********************}

\State {\bf upon} $\Phthreeindicator= 1$ {\bf do}:            \quad    \emph{// For some application in \cite{ChenOciorMVBA:24}, $\Phthreeindicator$ may be set $\Phthreeindicator\gets 1$ by receiving  a binary  input  $\ABAoutput =1$  (other than   message  input) } 
\Indent  
\If { $\COOLBUA$ has delivered $[\Me^{(i)}, \Ry_{i}^{[2]},   \Vr_i ]$ with $\Ry_i^{[2]} = 1$}         
	\State $\Output$   $\Me^{(i)}$ and $\terminate$    
		 
\Else
 
			\State  $\wait$ until $\COOLBUA$ delivering at least  $\networkfaultsizet +1$ $\ltuple\SYMBOL, \IDCOOL, (y_i^{(j)},  *) \rtuple$, $\forall j \in   \Ss_1^{[2]}$, for the same    $y_i^{(j)} =y^{\star}$,  for some   $y^{\star}$   
			\State  $y_i^{(i)} \gets y^{\star}$      \quad\quad\quad\quad \quad\quad\quad\quad\quad\quad\quad\quad\quad\quad\quad\quad   \quad\quad\quad\quad\quad\quad\quad\quad \quad\quad \emph{// update coded symbol based on  majority rule}  
	\State   $\send$   $\ltuple \CORRECTSYMBOL, \IDCOOL, y_i^{(i)}  \rtuple$  to  all nodes     
			\State  $\wait$ until  $\OECSIFinal=1$ 	
			\State $\Output$   $\Me^{(i)}$ and $\terminate$     
	
\EndIf

\EndIndent

\State {\bf upon} receiving   $\ltuple\CORRECTSYMBOL, \IDCOOL, y_j^{(j)} \rtuple$  from  Node~$j$ for the first time,   $j\notin \OECCorrectSymbolSet$, and $\OECSIFinal=0$   {\bf do}:  
\Indent  
	\State $\OECCorrectSymbolSet[j] \gets y_j^{(j)}$    
 	\If  { $|\OECCorrectSymbolSet|\geq  \kencode + \networkfaultsizet  $}     \label{line:RBAOECbegin}    \quad    \quad\quad \quad \quad\quad\quad \quad    \quad\quad \quad \quad\quad\quad \quad   \quad\quad \quad \quad\quad\quad \quad   \quad\quad \quad \quad\quad\quad \quad    \emph{//   online error correcting   }  
			\State   $\MVBAOutputMsg  \gets \ECCDec(n,  \kencode , \OECCorrectSymbolSet)$	     
			\State  $[y_{1}, y_{2}, \cdots, y_{n}] \gets \ECCEnc (n,  \kencode, \MVBAOutputMsg)$ 
    			\If {at least $\kencode + \networkfaultsizet$ symbols in $[y_{1}, y_{2}, \cdots, y_{n}]$ match with  those in $\OECCorrectSymbolSet$}
 				set  $ \Me^{(i)} \gets \MVBAOutputMsg$ and  $\OECSIFinal\gets 1$     \label{line:RBAOECend}	 
    			\EndIf    

	\EndIf   
		     
\EndIndent

\State {\bf upon} $\COOLBUA$ having  delivered  $\ltuple\SYMBOL, \IDCOOL, (*, y_j^{(j)})\rtuple$  and $\Ss_1^{[2]}$  such that   $j\in \Ss_1^{[2]}$, and $j\notin \OECCorrectSymbolSet$, and  $\OECSIFinal=0$   {\bf do}:    
\Indent  
	\State $\OECCorrectSymbolSet[j] \gets y_j^{(j)}$     
	\State run the OEC steps as in Lines~\ref{line:RBAOECbegin}-\ref{line:RBAOECend}       
\EndIndent

\end{algorithmic}
\end{algorithm}

\begin{algorithm}
\caption{$\OciorRBC$ protocol with identifier $ \IDCOOL $.  Code is shown for    Node~$\thisnodeindex$ for $ \thisnodeindex \in [n]$. }  \label{algm:OciorRBC}  
\begin{algorithmic}[1]

\footnotesize

\Statex \emph{// ** $\OciorRBC$ is a balanced $\RBC$ protocol that balances communication between the leader and the other nodes. **}

\Statex \emph{// ** Without balancing, in the initial phase the leader simply broadcasts the entire message to each node (see Lines~\ref{algm:OciorRBAWBCbegin}-\ref{algm:OciorRBAWBCEnd}). **}

  \Statex
\State  Initially set   $\kencode \gets      \networkfaultsizet / 3; \Me_{i}\gets \defaultvalue;   \OECSI \gets 0;    \OECsymbolset \gets  \{\};       \Phthreeindicator\gets 0 $

 \vspace{5pt}

  \Statex   \emph{//   ********************  Initial Phase  (With Balanced  Communication) ********************}

\State {\bf upon} receiving a non-empty  message  input $\wv$, and if this node is the leader {\bf do}:   
\Indent  
		\State  $[\OECsymbol_1, \OECsymbol_2, \cdots, \OECsymbol_{n}]\gets \ECCEnc(n, k, \wv)$
		\State   $\send$ $\ltuple   \LEADER, \IDCOOL,  \OECsymbol_{j} \rtuple $ to  Node~$j$,    $\forall j \in  [n]$       
 
\EndIndent

\State {\bf upon} receiving   $\ltuple   \LEADER, \IDCOOL,  \OECsymbol_{i} \rtuple $  from  the leader for the first time {\bf do}:  
\Indent  
		\State   $\send$ $\ltuple   \INITIAL, \IDCOOL,   \OECsymbol_{i} \rtuple $ to  all nodes             \quad\quad \quad \quad\quad\quad \quad   \quad\quad \quad \quad\quad\quad \quad   \quad\quad \quad \quad\quad\quad \quad \quad   \quad\quad \quad   \emph{// echo coded  symbol}   		
\EndIndent

\State {\bf upon} receiving message  $\ltuple   \INITIAL, \IDCOOL,  \OECsymbol_{j} \rtuple $ from Node~$j$  for the first time, and $\OECSI=0$   {\bf do}:   
\Indent  

		\State $\OECsymbolsetInitial[j] \gets \OECsymbol_{j}$	
 		\If  { $|\OECsymbolsetInitial|\geq  k + t  $}      \quad    \quad\quad \quad \quad\quad\quad \quad    \quad\quad \quad \quad\quad\quad \quad   \quad\quad \quad \quad\quad\quad \quad   \quad\quad \quad \quad\quad\quad \quad    \emph{//   online error correcting  (OEC)  }  
			\State   $\tilde{\wv}  \gets \ECCDec(n,  k, \OECsymbolsetInitial)$	
			\State  $ [\OECsymbol_{1}', \OECsymbol_{2}', \cdots, \OECsymbol_{n}'] \gets \ECCEnc (n,  k, \tilde{\wv})$ 
    			\If {at least $k + t$ symbols in $ [\OECsymbol_{1}', \OECsymbol_{2}', \cdots, \OECsymbol_{n}'] $ match with  those in $\OECsymbolsetInitial$, and $\tilde{\wv}$ is non-empty}
 				\State  $\Me_{i} \gets \tilde{\wv}; \OECSI\gets 1$   	 
    			\EndIf    
		\EndIf   
	    		
\EndIndent

 	  \vspace{5pt}
 
   \Statex   \emph{//   ********************  Initial Phase  (Only for the Case Without Balanced Communication) ********************}

\State {\bf upon} receiving a non-empty  message  input $\wv$, and if this node is the leader {\bf do}:    \label{algm:OciorRBAWBCbegin}    
\Indent  
		\State   $\send$ $\ltuple   \LEADERMESSAGE, \IDCOOL,  \wv \rtuple $ to  all nodes        
 
\EndIndent

\State {\bf upon} receiving   $\ltuple   \LEADERMESSAGE, \IDCOOL,  \wv \rtuple $  from  the leader for the first time {\bf do}:  
\Indent  
 				\State  $\Me_{i} \gets \tilde{\wv}$   	    \label{algm:OciorRBAWBCEnd}
\EndIndent

\Statex
  
  \Statex   \emph{//   ********************  $\OciorRBA$********************} 
	
\State {\bf upon} $\Me_{i}\neq \defaultvalue$ {\bf do}:

\Indent  

 \State    $\Pass$ $\Me_{i}$ into  $\OciorRBA$    as an input value
 
\EndIndent

\State {\bf upon}  delivery of  the output value $\Me^{(i)}$  from   $\OciorRBA$  {\bf do}:     
\Indent  
			\State $\Output$   $\Me^{(i)}$ and $\terminate$     
\EndIndent

\end{algorithmic}
\end{algorithm}

%

\begin{thebibliography}{10}
\providecommand{\url}[1]{#1}
\csname url@samestyle\endcsname
\providecommand{\newblock}{\relax}
\providecommand{\bibinfo}[2]{#2}
\providecommand{\BIBentrySTDinterwordspacing}{\spaceskip=0pt\relax}
\providecommand{\BIBentryALTinterwordstretchfactor}{4}
\providecommand{\BIBentryALTinterwordspacing}{\spaceskip=\fontdimen2\font plus
\BIBentryALTinterwordstretchfactor\fontdimen3\font minus
  \fontdimen4\font\relax}
\providecommand{\BIBforeignlanguage}[2]{{%
\expandafter\ifx\csname l@#1\endcsname\relax
\typeout{** WARNING: IEEEtran.bst: No hyphenation pattern has been}%
\typeout{** loaded for the language `#1'. Using the pattern for}%
\typeout{** the default language instead.}%
\else
\language=\csname l@#1\endcsname
\fi
#2}}
\providecommand{\BIBdecl}{\relax}
\BIBdecl

\bibitem{PSL:80}
M.~Pease, R.~Shostak, and L.~Lamport, ``Reaching agreement in the presence of
  faults,'' \emph{Journal of the ACM}, vol.~27, no.~2, pp. 228--234, Apr. 1980.

\bibitem{LSP:82}
L.~Lamport, R.~Shostak, and M.~Pease, ``The {Byzantine} generals problem,''
  \emph{ACM Transactions on Programming Languages and Systems (TOPLAS)},
  vol.~4, no.~3, pp. 382--401, Jul. 1982.

\bibitem{ChenDISC:21}
J.~Chen, ``Optimal error-free multi-valued {Byzantine} agreement,'' in
  \emph{International Symposium on Distributed Computing (DISC)}, Oct. 2021.

\bibitem{Chen:2020arxiv}
------, ``Fundamental limits of {Byzantine} agreement,'' 2020, available on
  arXiv: https://arxiv.org/pdf/2009.10965.pdf.

\bibitem{ChenOciorCOOL:24}
------, ``{OciorCOOL: Faster Byzantine agreement and reliable broadcast},''
  Sep. 2024, available on arXiv: https://arxiv.org/abs/2409.06008.

\bibitem{ChenOciorMVBA:24}
------, ``{OciorMVBA: Near-optimal error-free asynchronous MVBA},'' Dec. 2024,
  available on arXiv: https://arxiv.org/abs/2501.00214.

\bibitem{ChenOciorABA:25}
------, ``{OciorABA: Improved error-free asynchronous Byzantine agreement via
  partial vector agreement},'' Jan. 2025, available on arXiv:
  https://arxiv.org/abs/2501.11788.

\bibitem{ChenOcior:25}
------, ``{Ocior: Ultra-fast asynchronous leaderless consensus with two-round
  finality, linear overhead, and adaptive security},'' Sep. 2025, available on
  arXiv: https://arxiv.org/abs/2509.01118.

\bibitem{LCabaISIT:21}
F.~Li and J.~Chen, ``Communication-efficient signature-free asynchronous
  {B}yzantine agreement,'' in \emph{Proc. {IEEE} Int. Symp. Inf. Theory
  {(ISIT)}}, Jul. 2021.

\bibitem{ZLC:23}
J.~Zhu, F.~Li, and J.~Chen, ``Communication-efficient and error-free gradecast
  with optimal resilience,'' in \emph{Proc. {IEEE} Int. Symp. Inf. Theory
  {(ISIT)}}, Jun. 2023, pp. 108--113.

\bibitem{FH:06}
M.~Fitzi and M.~Hirt, ``Optimally efficient multi-valued {Byzantine}
  agreement,'' in \emph{Proceedings of the ACM Symposium on Principles of
  Distributed Computing (PODC)}, Jul. 2006, pp. 163--168.

\bibitem{LV:11}
G.~Liang and N.~Vaidya, ``Error-free multi-valued consensus with {Byzantine}
  failures,'' in \emph{Proceedings of the ACM Symposium on Principles of
  Distributed Computing (PODC)}, Jun. 2011, pp. 11--20.

\bibitem{GP:20}
C.~Ganesh and A.~Patra, ``Optimal extension protocols for {Byzantine} broadcast
  and agreement,'' in \emph{Distributed Computing}, Jul. 2020.

\bibitem{LDK:20}
A.~Loveless, R.~Dreslinski, and B.~Kasikci, ``Optimal and error-free
  multi-valued {Byzantine} consensus through parallel execution,'' 2020,
  available on: https://eprint.iacr.org/2020/322.

\bibitem{NRSVX:20}
K.~Nayak, L.~Ren, E.~Shi, N.~Vaidya, and Z.~Xiang, ``Improved extension
  protocols for {Byzantine} broadcast and agreement,'' in \emph{International
  Symposium on Distributed Computing (DISC)}, Oct. 2020.

\bibitem{Patra:11}
A.~Patra, ``Error-free multi-valued broadcast and {Byzantine} agreement with
  optimal communication complexity,'' in \emph{International Conference on
  Principles of Distributed Systems (OPODIS)}, 2011, pp. 34--49.

\bibitem{CT:05}
C.~Cachin and S.~Tessaro, ``Asynchronous verifiable information dispersal,'' in
  \emph{IEEE Symposium on Reliable Distributed Systems (SRDS)}, Oct. 2005.

\bibitem{CDG+:24}
P.~Civit, M.~A. Dzulfikar, S.~Gilbert, R.~Guerraoui, J.~Komatovic,
  M.~Vidigueira, and I.~Zablotchi, ``Efficient signature-free validated
  agreement,'' in \emph{International Symposium on Distributed Computing
  (DISC)}, vol. 319, Oct. 2024, pp. 14:1--14:23.

\bibitem{EW:25}
M.~Mizrahi~Erbes and R.~Wattenhofer, ``Brief announcement: Extending
  asynchronous {Byzantine} agreement with crusader agreement,'' in
  \emph{Proceedings of the ACM Symposium on Principles of Distributed Computing
  (PODC)}, Jun. 2025, pp. 50--53.

\bibitem{AA:25}
I.~Abraham and G.~Asharov, ``{ABEL: Perfect asynchronous Byzantine extension
  from list-decoding},'' in \emph{International Symposium on Distributed
  Computing (DISC)}, Oct. 2025, pp. 1:1--1:20.

\bibitem{Sudan:97}
M.~Sudan, ``Decoding of reed solomon codes beyond the error-correction bound,''
  \emph{Journal of Complexity}, vol.~13, no.~1, pp. 180--193, Mar. 1997.

\bibitem{GW:13}
V.~Guruswami and C.~Wang, ``Linear-algebraic list decoding for variants of
  reed?solomon codes,'' \emph{IEEE Trans. Inf. Theory}, vol.~59, no.~6, pp.
  3257--3268, 2013.

\bibitem{RS:60}
I.~Reed and G.~Solomon, ``Polynomial codes over certain finite fields,''
  \emph{Journal of the Society for Industrial and Applied Mathematics}, vol.~8,
  no.~2, pp. 300--304, Jun. 1960.

\bibitem{SS:96}
M.~Sipser and D.~Spielman, ``Expander codes,'' \emph{IEEE Trans. Inf. Theory},
  vol.~42, no.~6, pp. 1710--1722, Nov. 1996.

\bibitem{BCG:93}
M.~Ben-Or, R.~Canetti, and O.~Goldreich, ``Asynchronous secure computation,''
  in \emph{Proceedings of the Twenty-Fifth Annual ACM Symposium on Theory of
  Computing}, 1993, pp. 52--61.

\bibitem{Berlekamp:68}
E.~Berlekamp, ``Nonbinary {BCH} decoding (abstr.),'' \emph{IEEE Trans. Inf.
  Theory}, vol.~14, no.~2, pp. 242--242, Mar. 1968.

\bibitem{Gao:03}
S.~Gao, ``A new algorithm for decoding {Reed-Solomon} codes,'' in
  \emph{Communications, Information and Network Security}.\hskip 1em plus 0.5em
  minus 0.4em\relax Springer, 2003, pp. 55--68.

\bibitem{roth:06}
R.~Roth, \emph{Introduction to coding theory}.\hskip 1em plus 0.5em minus
  0.4em\relax Cambridge University Press, 2006.

\end{thebibliography}


\end{document}